\documentclass{article}
\usepackage{graphicx,amsmath,amssymb,tikz,amsthm,bm} 
\usepackage{arydshln}
\usepackage{mathtools}
\usepackage[linktocpage=true,
  colorlinks=true, 
  pdfborder={0 0 0},
  linkcolor=blue,
  citecolor=red,
  filecolor=yellow,
  urlcolor=blue,
  bookmarks,
  pdfauthor={},
]{hyperref}
\usetikzlibrary{quantikz}
\usepackage{authblk}

\newtheorem{theorem}{Theorem}

\newtheorem{example}{Example}

\title{Approximate real-time evolution operator for potential with\\ one ancillary qubit and application to first-quantized\\ Hamiltonian simulation}

\author[1,2]{Xinchi Huang\footnote{Email: kkou@quemix.com}}
\author[1,2]{Taichi Kosugi}
\author[1,2]{Hirofumi Nishi}
\author[1,2,3,4]{Yu-ichiro Matsushita}
\affil[1]{Department of Physics, The University of Tokyo, Tokyo 113-0033, Japan}
\affil[2]{Quemix Inc., Taiyo Life Nihombashi Building, 2-11-2, Nihombashi Chuo-ku, Tokyo 103-0027, Japan}
\affil[3]{Quantum Materials and Applications Research Center, National Institutes for Quantum Science and Technology (QST), 2-12-1 Ookayama, Meguro-ku, Tokyo 152-8550, Japan}
\affil[4]{Laboratory for Materials and Structures, Institute of Innovative Research, Tokyo Institute of Technology, Yokohama 226-8503, Japan}

\date{}

\begin{document}

\maketitle

\begin{abstract}
In this article, we compare the methods implementing the real-time evolution operator generated by a unitary diagonal matrix where its entries obey a known underlying real function. 
When the size of the unitary diagonal matrix is small, a well-known method based on Walsh operators gives a good and precise implementation. In contrast, as the number of qubits grows, the precise one uses exponentially increasing resources, and we need an efficient implementation based on suitable approximate functions. Using piecewise polynomial approximation of the function, we summarize the methods with different polynomial degrees. 
Moreover, we obtain the overheads of gate count for different methods concerning the error bound and grid parameter (number of qubits). This enables us to analytically find a relatively good method as long as the underlying function, the error bound, and the grid parameter are given. 
This study contributes to the problem of encoding a known function in the phase factor, which plays a crucial role in many quantum algorithms/subroutines. In particular, we apply our methods to implement the real-time evolution operator for the potential part in the first-quantized Hamiltonian simulation and estimate the resources (gate count and ancillary qubits) regarding the error bound, which indicates that the error coming from the approximation of the potential function is not negligible compared to the error from the Trotter-Suzuki formula. 
\end{abstract}
\section{Introduction}

Let $n\in \mathbb{N}$ be a given integer and $N = 2^n$.  
We define a diagonal matrix $V_N\in \mathbb{R}^{N \times N}$ in the following bra-ket notation:
\begin{align*}
V_N = \sum_{j=0}^{N-1} v_j \ket{j}\bra{j}, 
\end{align*}
and aim at finding a quantum circuit for the real-time evolution of this diagonal operator, i.e., $e^{-\mathrm{i}V_N \Delta\tau}$ where $\Delta\tau>0$ is a given small time step. 
Such a quantum circuit plays a fundamental role in the simulation of first-quantized Hamiltonian \cite{Lloyd1996, Abrams.1997, Zalka1998, Wiesner1996, Kassal.2008} as well as subroutines of many quantum circuits/algorithms, e.g., ground state preparation by using the imaginary-time evolution \cite{Gingrich.2004, Terashima.2005, McArdle.2019, Motta.2020, Kosugi.2022}. 
Considering the application to Hamiltonian simulation, we assume in this paper that the entries $v_j$, $j=0,1,\ldots, N-1$ take values of some given potential function at spatial uniform grid points.  
The implementation of a given unitary diagonal matrix is no longer an easy task when $n$ becomes large. In the following context, we summarize several methods in previous papers.

\noindent \underline{Walsh operators in sequency order}
Based on the gate set $\{\mathrm{CNOT}, R_z\}$, it is a well-known result that such a diagonal unitary operation on $n$ qubits can be precisely implemented using $2^n$ multifold $z$ rotation gates, corresponding to Walsh operators, whose rotation angles are proportional to the Hadamard transform of the given diagonal components \cite{Schuch2002, Schuch.2003, BM2004}. 
Later \cite{Welch.2014} gave an optimal circuit in CNOT count using sequency order based on Gray code \cite{G53, Beauchamp1984}, and \cite{Zhang.2022} further reduced the circuit depth by commuting and parallelizing quantum gates in a collective manner. 
Moreover, the circuit for the diagonal unitary matrix with reflection symmetry was discussed in \cite{HKNM24}, and another half reduction was achieved due to the symmetry.  
Although such optimal circuits reduce the depth up to a factor 2 or 4 without any ancillary qubit, the total gate count/circuit depth remains the order of $\mathcal{O}(2^n)$, which is still a tough task for large $n$.
As for a phase-sparse diagonal unitary matrix, the authors in \cite{Beer.2016, Welch.2016} showed that such exponential growth could be reduced to $\mathcal{O}(k n^2)$ where $k$ is the number of distinct values in the diagonal. 

\noindent \underline{Polynomial phase gate}
Based on multi-controlled phase gates and quantum comparators, there is a way to approximate the given potential function by piecewise polynomials and directly implement the approximate functions in the phase factor. 
Although this idea has already been proposed in the first decade of this century \cite{Benenti.2008}, it is not intensively discussed to the authors' best knowledge. 
Owing to the polynomial approximation, the total gate count is of polynomial order. Moreover, only one ancillary qubit is needed for storing the results of quantum comparators. 
Recently, \cite{Ollitrault.2020, Kosugi.2023} discussed the application of such polynomial phase gates to first-quantized Hamiltonian simulation, while the detailed gate count under a given precision was not provided. 

\noindent \underline{Arithmetic operations and phase kickback}
An alternative implementation of the piecewise polynomial approximation in the phase factor is to introduce several ancilla registers to store the function/intermediate results, calculate the desired function in an ancilla register, and then kick back the values in the register to the phase factor. 
In \cite{Kassal.2008, Jones.2012}, the authors suggested using arithmetic operations to calculate the approximate polynomial functions in the ancilla register. 
Recently, \cite{Babbush.2018, Haner.2018, Sanders.2020} provided detailed circuits on the loading of polynomial coefficients in ancilla registers by quantum read-only memory (QROM) \cite{Babbush.2018}, and on the parallel implementation of the polynomials by introducing a label/selection register. 
In comparison to the above polynomial phase gate, this method requires more ancillary qubits to conduct the arithmetic operations and to store the intermediate results. 

This paper is mainly devoted to the insight comparison between different methods, including the modified previous ones and a novel one based on a linear interpolated unitary diagonal matrix. The first contribution is the asymptotic resource estimation concerning the grid parameter and the desired precision (the error bound). 
Taking the limitation of qubit number for the early quantum computers into account, then we discuss the different methods with at most one ancillary qubit in detail and provide a practical way on how to choose an efficient quantum circuit for a given potential function under desired precision. 
Finally, using the above results, the gate count and the ancilla count are evaluated for the problem of first-quantized Hamiltonian simulation for a particle system. 

\section{Methods}
\label{sec:methods}

Recall that $N=2^n$ for a given grid parameter $n\in \mathbb{N}$. For $L>0$, we introduce a uniform division of the interval $[0,L]$ as $\Pi$: 
$0 = x_0 < x_1 < \cdots < x_N = L$, where $x_j = jL/N$, $j=0,1,\ldots,N$. For a given sufficiently smooth function $V$, we aim at approximate circuits to implement the following unitary diagonal matrix: 
$$
e^{-\mathrm{i}V_N\Delta\tau} = \sum_{j=0}^{N-1}e^{-\mathrm{i}v_j\Delta\tau} \ket{j}\bra{j},
$$
where $v_j = V(x_j)$, $j=0,1,\ldots,N-1$. 
As we can put the time step factor $\Delta\tau$ into $V_N$, without loss of generality, here we take $\Delta\tau=1$ for simplicity. 
In the following contexts, we review previous methods and provide a new method for the above purpose. 

\subsection{Sequency-ordered Walsh operator (WAL)}
\label{subsec:method1}

Under the gate set $\{\mathrm{CNOT}, R_z\}$, the exact implementation of the unitary diagonal matrices can be done by using sequency-ordered Walsh operators \cite{Schuch2002, Schuch.2003, BM2004, Welch.2014, Zhang.2022}. The optimized circuit in this direction is given by \cite{Zhang.2022} (see Appendix \ref{subsec:appA1} for more details). 
The gate count/depth is scaled as $2^n$ with relatively small pre-factor $1$ for the depth and pre-factor $2$ for the gate count, which indicates that such a method is a good choice for small grid parameter $n$ \cite{Sornborger.2012}. 
For a large grid parameter $n$, as suggested by Welch et al. in \cite{Welch.2014}, one can choose a coarse-graining parameter $m_0<n$ and apply the sequency-ordered Walsh operators only to the top $m_0$ qubits to derive an approximate operation of the original unitary operation. 
If we denote the circuit for the unitary diagonal operation on $k$ qubits in \cite{Zhang.2022} by $U_{\text{D}}(V,k)$, then the quantum circuit of WAL is described in Fig.~\ref{met:Fig1}, where $m=m_0$ depends on the potential function and the precision. 
\begin{figure}
\centering
\resizebox{10cm}{!}{
\includegraphics[keepaspectratio]{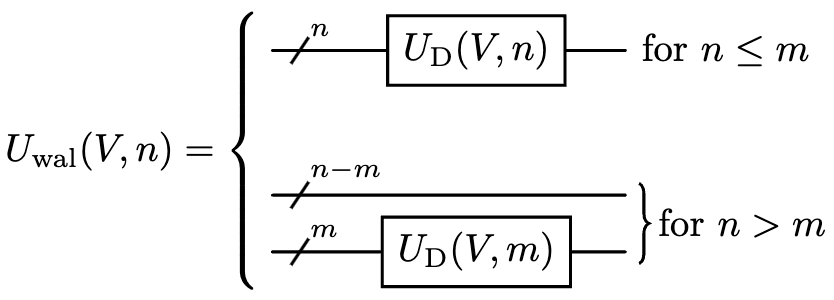}
}
\caption{Quantum circuit of WAL based on sequency-ordered Walsh operators. Throughout this paper, we locate the top (most significant) qubits in the lower part of the circuits. }
\label{met:Fig1}
\end{figure}
For the suitable choice of $m_0$, we refer to Appendix \ref{subsec:appA1}. 
Intuitively, if $n$ is strictly larger than $m_0$, then this method yields a quantum circuit corresponding to a uniform piecewise constant approximation of the given potential function. 

\subsection{Linear interpolated unitary diagonal matrix (LIU)}

We propose a novel method by constructing a linear interpolated unitary diagonal matrix from an original unitary diagonal matrix. Intuitively, the idea is to obtain a circuit for a uniform piecewise linear approximation. In other words, we first calculate the values at $2^{m_1}$ ($m_1<n$) coarse grid points and then apply linear interpolation at fine grid points.  
Let $1\le m_1<n$ be an integer describing the parameter of a coarse division and let $M=2^{m_1}$. We assume the access to a diagonal unitary operation and its fractional binary power:
\begin{align*}
U_M := \sum_{k=0}^{M-1} \exp\left(-\mathrm{i}u_{k}\right) \ket{k}\bra{k}, \quad 
U_M^{(j)} := 
\sum_{k=0}^{M-1} \exp\left(-\mathrm{i}u_{k}/2^j\right) \ket{k}\bra{k}, \ j=0,1,\ldots, n-m_1,
\end{align*}
where we set $u_k = v_{kN/M}$, $k=0,1,\ldots,M-1$. In a general case, the fractional binary power of a given diagonal unitary operation can be derived using the quantum circuit in \cite{Sheridan.2009}. 
By applying an increment gate defined by
\begin{align*}
U_{\text{+1}} \ket{k}_{m_1} := \ket{k+1}_{m_1},\ k=0,1,\ldots,M-1,
\end{align*}
we introduce a phase difference diagonal operation $W_{j}$: 
$$
W_{j} \ket{k}_{m_1} = U_{\text{+1}}^\dag U_M^{(j)} U_{\text{+1}} U_M^{(j)\dag} \ket{k}_{m_1}
= e^{-\mathrm{i}(u_{k+1}-u_{k})/2^{j}} \ket{k}_{m_1}, \ k=0,1,\ldots,M-1,\, j = 0,1,\ldots, n-m_1.
$$
Here, we set $u_{M} := u_0$. 
\begin{figure}
\centering
\resizebox{10cm}{!}{
\includegraphics[keepaspectratio]{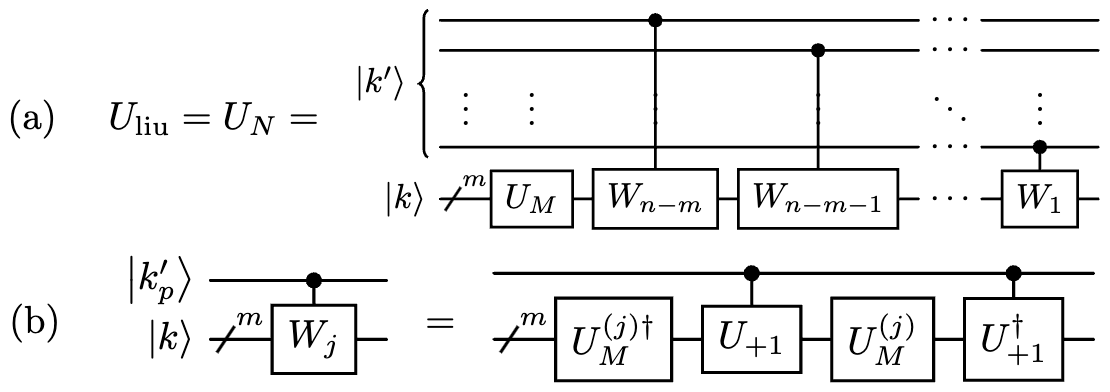}
}
\caption{(a) Proposed quantum circuit of LIU based on a linear interpolated diagonal unitary matrix. The controlled $W_j$ gate with control qubit $\ket{k_p^\prime}$ and target register $\ket{k}_m$ is denoted by $\mathrm{C}W_j^{(p)}$. 
(b) Quantum circuit of $\mathrm{C}W_j^{(p)}$ using two $U_M^{(j)}$, a controlled increment gate $U_{+1}$, and its inverse.}
\label{met:Fig2}
\end{figure}
By denoting $\ket{k^\prime}_{n-m_1} = \ket{k_{n-m_1-1}^\prime}\cdots\ket{k_{0}^\prime}$ and $\ket{j}_n = \ket{k}_{m_1} \otimes \ket{k^\prime}_{n-m_1}$ for $j=0,1,\ldots, N-1$, we employ the controlled $W_j$ gate (see Fig.~\ref{met:Fig2}(b)) and the diagonal unitary operation $U_M^{(j)}$ to obtain 
\begin{align*}
U_N\ket{j}_n &= \mathrm{C}W_1^{(n-m_1-1)} \cdots \mathrm{C}W_{n-m_1-1}^{(1)} \mathrm{C}W_{n-m_1}^{(0)}(U_M\otimes I^{\otimes (n-m_1)}) \left(\ket{k}_{m_1} \otimes \ket{k^\prime}_{n-m_1}\right)\\
&= \mathrm{C}W_1^{(n-m_1-1)} \cdots \mathrm{C}W_{n-m_1-1}^{(1)} \mathrm{C}W_{n-m_1}^{(0)} \left(\mathrm{exp}\left(-\mathrm{i} u_{k}\right)\ket{k}_{m_1} \otimes \ket{k^\prime}_{n-m_1}\right)\\
&= \exp\left(-\mathrm{i}\left(u_{k}+\sum_{p=0}^{n-m_1-1}k_{p}^\prime 2^{p}(u_{k+1}-u_{k})/2^{n-m_1}\right)\right) \left(\ket{k}_{m_1} \otimes \ket{k^\prime}_{n-m_1}\right)\\
&= e^{-\mathrm{i}\tilde v_j} \ket{j}_n.
\end{align*}
Here, we set $\tilde v_j := u_k + k^\prime(u_{k+1}-u_k)/2^{n-m_1}$, which indicates that the circuit in Fig.~\ref{met:Fig2} yields a diagonal unitary operation on $n$ qubits whose diagonal entries are the linear interpolation of the diagonal entries of $U_M$. 
In particular, we find $\tilde v_{\ell N/M} = u_{\ell} = v_{\ell N/M}$, $\ell=0,\ldots,M-1$. Then, $U_N$ can be interpreted as a linear interpolated matrix of the original matrix $U_M$. 
Since we know every entry $u_k$ using the underlying function $V$, we can directly implement the diagonal unitary operation $U_M^{(j)}$ by $U_{\text{D}}(V/2^j,m_1)$. 
The detailed gate implementation and the choice of parameter $m_1$ are provided in Appendix \ref{subsec:appA2}. 

Our proposed method LIU works well if the potential function is periodic or localized such that the difference between the values at two endpoints is negligible. For a general function, we do not necessarily have $u_M = u_0$, which implies that there may be an additional error in the last sub-interval of the division, and we may need extra multi-controlled phase gates to refine the circuit. 
An alternative way is to extend the function from $[0,L]$ to $[0,2L]$ by $V(x):= V(2L-x)$ for $x\in [L,2L]$. Then, we derive a periodic function that can be dealt with. To maintain the desired precision, we need one zero-initialized ancillary qubit. In detail, we apply $U_M^{(j)}$ to the top $m_1$ qubits and the ancillary qubit, and the controlled increment gates are also targeted at these $m_1+1$ qubits. The gate count discussed in Appendix \ref{subsec:appA2} remains the same formula, but one needs to replace $n$ and $m_1$ by $n+1$ and $m_1+1$, respectively. A third way for modification relies on the direct implementation of the controlled $W_j$ operation, given by the controlled unitary diagonal operation, which consists of $2^{m_1}-2$ CNOT gates and $2^{m_1}-1$ controlled phase gates. We will call this modified method mLIU in the following contexts. 
To avoid confusion, we mention a quantum interpolation algorithm for the quantum states, which is useful for preparing ``finer" initial quantum states \cite{Ramos2022, Moosa.2023}. However, such an interpolation circuit does not work for our purpose since it will change the diagonality of the operation. 

\subsection{Phase gate for piecewise-defined polynomial (PPP)}
\label{subsec:met3}

It is known that any function can be approximated by some suitable piecewise polynomial function for which its division and polynomial coefficients can be derived by classical algorithms. 
Then, the problem of finding an efficient approximate circuit for the evolution operator of the potential boils down to the efficient implementation of a given piecewise polynomial function in the phase factor. 
Employing a quantum comparator to distinguish the different sub-intervals of the piecewise polynomial, one solution is to use the polynomial phase gate for a given polynomial function $f$ as follows \cite{Benenti.2008, Kosugi.2023}:
\begin{align*}
U_{\text{ph}}[f] = \sum_{j=0}^{N-1} \mathrm{exp}\left(-\mathrm{i}f(x_j)\right) \ket{j}\bra{j}.
\end{align*} 
A single polynomial phase gate is constructed by several multi-controlled phase gates \cite{Kosugi.2023} for which we provide more details in Appendix \ref{subsec:appA3} for completeness. 

Here, we suggest a modified way to determine the piecewise polynomial approximation to a given function, where we introduce a coarse-graining parameter $m\le n$ and let $M=2^m$. 
We choose $f(\tilde x_{\ell}) = V(\tilde x_{\ell})$ with some given knots $\tilde x_\ell \in \{x_{kN/M}\}_{k=0}^{M-1}$, $\ell=0,1,\ldots,\tilde M$ and $\tilde x_0 = x_0$, $\tilde x_{\tilde M} = x_N$. Since $\tilde M+1$ is the number of knots, $\tilde M\le M$ is the number of polynomials. 
Then, we can derive a piecewise polynomial function $\tilde V$ (using the values at these $\tilde M+1$ knots) by determining the polynomial coefficients with spline interpolation. 
In Appendix \ref{subsec:appA3}, we give an analytic estimation of the parameters $m,\tilde M$ so that $\max_{x\in [0, L]}|\tilde V(x)-V(x)|\le \delta$ for a given precision $\delta>0$, and we also provide practical classical algorithms for finding these parameters numerically with possibly varying degrees for each interval. 
As long as we have prepared the piecewise polynomial approximation: 
\begin{equation}
\label{analy:eq-PPA}
\tilde V(x) = \sum_{\ell=0}^{\tilde M-1} \chi_{[\tilde x_{\ell}, \tilde x_{\ell+1})}(x) f_{\ell}(x), \quad x\in [0,L],
\end{equation}
where $\chi_{A}$ is the characteristic function whose value is $1$ for $x\in A$ and $0$ otherwise, we realize the approximate diagonal operation 
\begin{align*}
e^{-\mathrm{i}\tilde V_N} = \sum_{j=0}^{N-1} e^{-\mathrm{i}\tilde V(x_j)}\ket{j}\bra{j},
\end{align*}
by employing phase gates for polynomial functions and quantum comparators as shown in Fig.~\ref{met:Fig3}. 
A quantum comparator is usually regarded as a circuit to compare the integers encoded in two quantum registers \cite{BKS.2018, Li.2020}. In this article, to minimize the number of qubits, we suggest applying the quantum comparator for a given integer in \cite{Yuan.2023}: $\text{COMP($k$)}\ket{0}\otimes \ket{j} := \ket{j< k} \otimes \ket{j}$, which requires only one ancillary qubit for storing the comparison result. 
\begin{figure}
\centering
\resizebox{15cm}{!}{
\includegraphics[keepaspectratio]{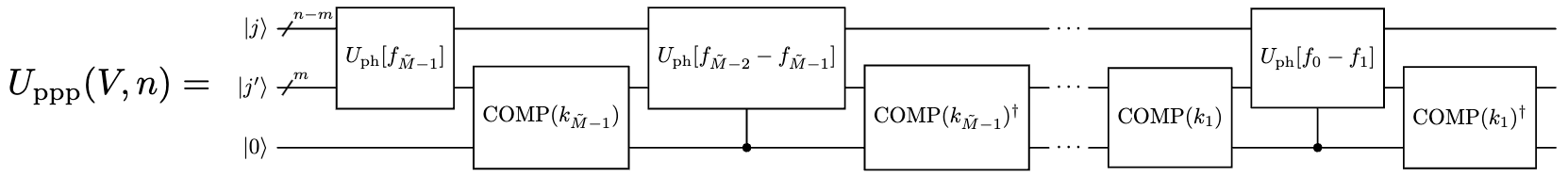}
}
\caption{Modified quantum circuit of PPP to implement a real-time evolution operator of an approximate potential by a piecewise polynomial function. Here, one can use the quantum comparator with a given integer $k_\ell = M\tilde x_\ell/L$ in \cite{Yuan.2023}. The difference from \cite{Ollitrault.2020, Kosugi.2023} is that we apply the comparator only on the most significant $m$ qubits of the register, and thus, we reduce the number of gates for the comparators if $m<n$. } 
\label{met:Fig3}
\end{figure}
We discuss the detailed gate count and depth in Appendix \ref{subsec:appA3}, and here we mention that the above circuit uses one polynomial phase gate, $\tilde M-1$ controlled polynomial phase gates, and $2(\tilde M-1)$ quantum comparators, including the part of uncomputation so that the ancillary qubit can be reused safely. 

\subsection{Arithmetic operations and phase kickback (APK)}
\label{subsec:met4}

In the above circuits for WAL, LIU, and PPP, we directly implement the approximation $e^{-\mathrm{i}\tilde V(x_j)}$ in the phase factor:
$$
\ket{j} \mapsto e^{-\mathrm{i}\tilde V(x_j)}\ket{j}, \quad j=0,1,\ldots, N-1,
$$
using the polynomial phase gates. 
On the other hand, there is another way \cite{Kassal.2008, Babbush.2018, Haner.2018, Sanders.2020} to achieve it by introducing some auxiliary registers and using arithmetic operations. 
Let the piecewise polynomial approximation $\tilde V$ be given by Eq.~\eqref{analy:eq-PPA}. Then, the procedure is as follows:
\begin{enumerate}
\item Introduce and initialize a label/selection register using the system register \cite[Figure 1]{Haner.2018}:
$$
\ket{j}\ket{0} \mapsto \ket{j}\ket{\ell(j)}, 
$$
where $\ell(j)$ denotes the index of the interval for each point $x_j$. 

\item Introduce a register for data and input the piecewise polynomial coefficients $d_\ell$ into the register by quantum read-only memory (QROM) \cite{Babbush.2018, Sanders.2020}: 
$$
\ket{j}\ket{\ell(j)}\ket{0} \mapsto \ket{j}\ket{\ell(j)}\ket{d_{\ell(j)}}. 
$$

\item Introduce a register for the temporary result and calculate the approximate polynomial function for each interval by applying some arithmetic operations (multiplications and additions) conditioned on the system register and the data register \cite{Haner.2018, Sanders.2020}: 
$$
\ket{j}\ket{\ell(j)}\ket{d_{\ell(j)}}\ket{0} \mapsto \ket{j}\ket{\ell(j)}\ket{d_{\ell(j)}}\ket{f_{\ell(j)}(x_j)}
= \ket{j}\ket{\ell(j)}\ket{d_{\ell(j)}}\ket{\tilde V(x_j)}. 
$$

\item Employ the phase kickback circuit \cite{Jones.2012}, which substitutes phase (rotation) gates by an addition on the result register and a suitably prepared phase kickback register: 
\begin{align*}
\ket{j}\ket{\ell(j)}\ket{d_{\ell(j)}}\ket{\tilde V(x_j)} \ket{0}
&\mapsto \ket{j}\ket{\ell(j)}\ket{d_{\ell(j)}}\ket{\tilde V(x_j)}\ket{\gamma(\xi)}\\
&\mapsto e^{-\mathrm{i}\xi \tilde V(x_j)}\ket{j}\ket{\ell(j)}\ket{d_{\ell(j)}}\ket{\tilde V(x_j)}\ket{\gamma(\xi)}. 
\end{align*}
Here, we take $\xi=1$, and we can take, e.g., $\xi=\Delta \tau$ for applications. 

\item Discard the unentangled ancillary qubits and do the uncomputation:
\begin{align*}
e^{-\mathrm{i} \tilde V(x_j)}\ket{j}\ket{\ell(j)}\ket{d_{\ell(j)}}\ket{\tilde V(x_j)}\ket{\gamma(\xi)}
&\mapsto e^{-\mathrm{i} \tilde V(x_j)}\ket{j}\ket{\ell(j)}\ket{d_{\ell(j)}}\ket{\tilde V(x_j)} \\
&\mapsto e^{-\mathrm{i} \tilde V(x_j)}\ket{j}\ket{0}\ket{0}\ket{0}. 
\end{align*}
Here, the phase kickback register can be discarded, while one needs uncomputation for the result, data, and label/selection registers. 

\end{enumerate}
This method is based on arithmetic operations which are usually implemented by CNOT and Toffoli gates. On the other hand, it uses many ancillary qubits, e.g., hundreds of ancillary qubits for even elementary functions \cite{Haner.2018}, which is unsuitable for early quantum computers. 
Although it is not the main concentration of this paper, we collect the possible implementations from previous works in Appendix \ref{subsec:appA4}, and in the next section, we compare this method with the others in their performance on asymptotic gate count regarding grid parameter $n$ and precision $\delta$. 

\section{Gate count and comparison}
\label{sec:comp}

We start by collecting the gate count/circuit depth for the methods in Sect.~\ref{sec:methods}. The results are obtained under the explicit implementation described in Appendix \ref{sec:appA}. 

For WAL, considering the approximate implementation with a coarse-graining parameter $m_0$, one needs $2^{\min\{m_0,n\}}-1$ phase gates and $2^{\min\{m_0,n\}}-2$ CNOT gates with depth $2^{\min\{m_0,n\}}$. In this section, we omit the global phase and do not distinguish phase (rotation) gates from $R_z$ gates due to their equivalence to a global phase. 
According to Appendix \ref{subsec:appA1}, to guarantee the desired precision $\delta>0$, we choose
\begin{equation}
\label{comp:m0}
m_0 := \left\lceil \log_2 \left(L\frac{\|V^{(1)}\|_\infty}{\delta}\right) \right\rceil,
\end{equation}
where $\|V^{(k)}\|_\infty$ denotes the maximum norm of the $k$-th derivative of $V$. 

For LIU, we use quantum Fourier transforms (QFTs) for the increment gate \cite{Yuan.2023} and introduce a parameter $m_1<n$. According to Appendix \ref{subsec:appA2}, we implement the desired unitary diagonal matrix with $4m_1(n-m_1)$ Hadamard gates, $2(2^{m_1}+3m_1^2-1)(n-m_1)+2^{m_1}-1$ phase gates, and $2(2^{m_1}+2m_1^2-2)(n-m_1)+2^{m_1}-2$ CNOT gates up to a global phase. Moreover, the depth under the gate set $\{\mathrm{H}, \mathrm{CNOT}, R_z\}$ is $2(2^{m_1}+16m_1-16)(n-m_1)+2^{m_1}$. On the other hand, since a controlled phase gate can be decomposed into 2 CNOT gates and 3 phase gates, mLIU needs $(3\cdot 2^{m_1}-3)(n-m_1)+2^{m_1}-1$ phase gates and $(3\cdot 2^{m_1}-4)(n-m_1)+2^{m_1}-2$ CNOT gates with depth $2(2^{m_1+1}+16m_1)(n-m_1)+2^{m_1+1}$.
For the desired precision $\delta>0$, we take 
\begin{equation}
\label{comp:m1}
m_1 :=\left\lceil \log_2 \left(L\sqrt{\frac{\|V^{(2)}\|_\infty}{8\delta}}\right) \right\rceil.
\end{equation}

For PPP with a given polynomial degree $p$, we introduce a parameter $m_p$. 
According to Appendix \ref{subsec:appA3}, we need $(8m_1+4)(\tilde M_1-1)$ Hadamard gates, $n+(3n+12m_1^2+4m_1+3)(\tilde M_1-1)$ phase gates, and $(8m_1^2+2n)(\tilde M_1-1)$ CNOT gates with depth $1+(4n+64m_1-32)(\tilde M_1-1)$ for $p=1$, and $(8m_2+4)(\tilde M_2-1)$ Hadamard gates, $n+3n(n-1)/2+(7n(n-1)/2+3n+12m_2^2+4m_2+3)(\tilde M_2-1)$ phase gates, and $n(n-1)+(4n(n-1)+2n+8m_2^2)(\tilde M_2-1)$ CNOT gates with depth 
$\mathcal{O}(\tilde M_2(n^2+m_2))$ for $p=2$. 
To guarantee the desired precision $\delta>0$, we take 
\begin{align*}
m_2 := \left\lceil\log_2 \left(L\left(\frac{2\|V^{(3)}\|_\infty}{81\delta}\right)^{\frac{1}{3}}\right) \right\rceil, \quad
\end{align*}
and 
\begin{align*}
\tilde M_2 \simeq \hat M_{2} := \left\lceil\int_0^L \left(\frac{2|V^{(3)}(x)|}{81\delta}\right)^{\frac{1}{3}} dx\right\rceil, \quad
\tilde M_1 \simeq \hat M_{1} := \left\lceil\int_0^L \left(\frac{|V^{(2)}(x)|}{8\delta}\right)^{\frac{1}{2}} dx\right\rceil.
\end{align*}
Here, $\hat M_{p}$ is the approximation of $\tilde M_p$ up to a pre-factor (see Appendices \ref{subsec:appA3}, \ref{subsec:appB2}), and we choose the constants $C_p$ in Eq.~\eqref{app:err-est} as $C_1=1/8$ and $C_2=2/81$, which correspond to the optimal constants for two-point Hermite interpolations \cite{Agarwal.1991}. By noting that $j$-controlled phase gate can be implemented by $\mathcal{O}(j^2)$ CNOT gates and single-qubit gates without any ancillary qubit 
in linear depth $\mathcal{O}(j)$ \cite{SP2022}, we need $\mathcal{O}((p^2n^p+m_p^2)\tilde M_p)$ CNOT gates and single-qubit gates with depth $\mathcal{O}((pn^p+m_p)\tilde M_p)$ for general $p\ge 3$. 
On the other hand, we have a comment on the special case that $\|V^{(p+1)}\|_\infty=0$ for some $p\in \mathbb{N}$, which means $V$ itself is a polynomial of degree smaller than $p$. According to the error estimates for spline interpolations, $m_p$ can be arbitrarily chosen (e.g., $m_p=0$), and we have $\tilde M_p=1$. 

According to Appendix \ref{subsec:appA4}, APK uses $\mathcal{O}\left((n+p)\delta^{-1/(p+1)}+p(p+\log (1/\delta))^2\right)$ Toffoli gates, $\mathcal{O}\left((n+p)\delta^{-1/(p+1)}+p(p+\log (1/\delta))^2\right)$ Clifford gates, and $\mathcal{O}\left((\log (1/\delta))^2\right)$ phase gates with 
$\mathcal{O}(1+p(p+\log (1/\delta)))$ initialized ancillary qubits.

\subsection{Asymptotic gate count}
\label{subsec:comp-1}

Although we cannot directly compare the methods since their implementations are based on different gate sets, we can derive their asymptotic dependence on the parameters, which also provides information on the features of these methods. The asymptotic gate count and ancilla count concerning polynomial degree $p$, precision $\delta$, and large grid parameter $n$ (so that gate count of WAL is $\mathcal{O}(\min\{2^n, \delta^{-1}\}=\mathcal{O}(\delta^{-1}))$) are given in Table \ref{tab:comp1} using the above estimations of the gate count. 
\begin{table}[htb]
\centering
\caption{Asymptotic gate count and ancilla count for different methods with respect to grid parameter $n$, polynomial degree $p$, and precision $\delta$. }
\label{tab:comp1}
\scalebox{0.8}[0.8]{
\begin{tabular}{l|cccc}
\hline
& WAL & LIU/mLIU & PPP ($p\ge 1$) & APK ($p\ge 1$) \\
\hline
Asymptotic gate count & $\mathcal{O}(\delta^{-1})$ & $\mathcal{O}(\delta^{-1/2}n)$ & $\mathcal{O}(\delta^{-1/(p+1)}(n^p+(\log 1/\delta)^2))$ & $\mathcal{O}((n+p)\delta^{-1/(p+1)}+p(p+\log 1/\delta)^2)$ \\
Ancilla count & 0 & 0/1 & 1 & $\mathcal{O}(1+p(p+\log 1/\delta))$ \\
\hline
\end{tabular}
}
\end{table}

For large fixed $n$, it is clear that PPP and APK both provide the best order of $\mathcal{\tilde O}(\delta^{-1/(p+1)}) = \mathcal{O}(\delta^{-1/(p+1)}\mathrm{poly}\log \delta^{-1})$ for $p\ge 2$. 
In particular, APK performs better than PPP in high-precision simulations at the cost of using $\mathcal{O}(\log 1/\delta)$ ancillary qubits since the constant before $\delta^{-1/(p+1)}$ is independent of $\delta$, and the dependence on $n$ is linear in APK while it is polynomial in PPP. 
This observation indicates the possible superiority of APK among the mentioned methods in the problem of high-precision simulations.  
However, as the near-term quantum devices do not afford such a large amount of logical qubits (hundreds of qubits due to \cite{Haner.2018}), we focus on WAL, LIU, and PPP, which employ at most one ancillary qubit and leave further discussion on APK to future work. 

Next, we compare WAL, LIU, and PPP in detail and list the asymptotic gate count and depth as $n\to \infty$ and $\delta\to 0$ in Table \ref{tab:comp2}.
\begin{table}[htb]
\centering
\caption{Required resources for different methods with respect to grid parameter $n$ and precision $\delta$. }
\label{tab:comp2}
\scalebox{0.9}[0.9]{
\begin{tabular}{l|cccc}
\hline
& WAL & LIU/mLIU & PPP ($p=1$) & PPP ($p\ge 2$) \\
\hline
CNOT count & $\mathcal{O}(\min\{2^n,\delta^{-1}\})$ & $\mathcal{O}(\delta^{-1/2}n)$ & $\mathcal{O}(\delta^{-1/2}(n+(\log 1/\delta)^2))$ & $\mathcal{O}(\delta^{-1/(p+1)}(n^p+(\log 1/\delta)^2))$ \\
$R_z$ count & $\mathcal{O}(\min\{2^n,\delta^{-1}\})$ & $\mathcal{O}(\delta^{-1/2}n)$ & $\mathcal{O}(\delta^{-1/2}(n+(\log 1/\delta)^2))$ & $\mathcal{O}(\delta^{-1/(p+1)}(n^p+(\log 1/\delta)^2))$ \\
Depth & $\mathcal{O}(\min\{2^n,\delta^{-1}\})$ & $\mathcal{O}(\delta^{-1/2}n)$ & $\mathcal{O}(\delta^{-1/2}(n+\log 1/\delta))$ & $\mathcal{O}(\delta^{-1/(p+1)}(n^p+\log 1/\delta))$ \\
Ancilla count & 0 & 0/1 & 1 & 1 \\
\hline
\end{tabular}
}
\end{table}
From Table \ref{tab:comp2}, we find three facts. 
First, focusing on the first two rows, the $R_z$ count has the same order as the CNOT count, which implies that it is sufficient to investigate the CNOT count when we discuss the asymptotic gate count. 
Second, focusing on the first and third rows for PPP, the depth has a slightly weaker dependence $\delta^{-1/(p+1)}\log 1/\delta$ than $\delta^{-1/(p+1)}(\log 1/\delta)^2$ for the CNOT count, but both the depth and the CNOT count have the order $\mathcal{\tilde O}(\delta^{-1/(p+1)})$ if we do not care the less donimant factor. 
Third, focusing on the second and third columns, LIU/mLIU has less asymptotic gate count/depth than PPP with $p=1$ concerning $\delta$. 
If $n$ is independent of $\delta$, then PPP with large $p$ provides less gate count for sufficiently small $\delta$. However, it is not rare that the grid parameter $n$ depends on the desired precision $\delta$ considering the discretization error. In the following context, we tell how to choose the method regarding different $n,\delta$ for given functions. 

\subsection{Case study for selected functions}
\label{subsec:3-2}

The asymptotic gate count may not make sense because the grid parameter $n$ and the desired precision $\delta$ are chosen according to the practical applications. Here, we compare WAL, LIU, and PPP for two given functions with varying parameters $n$, $\delta$. 

\begin{example}
\label{exa:1}
$V(x) = A/\sqrt{a^2+(x-L/2)^2}$, \quad $x\in [0,L]$. 
\end{example}
The first case is the modified Coulomb potential peaked at the center of the interval with parameters $A$ and $a$ describing its amplitude and ``singularity" near the peak, respectively. 
As an illustration, we choose $A=1$, $L=20$, and plot the CNOT count of the implementations by WAL, LIU, and PPP for a couple of precision $\delta$ and parameter $a$. 
For PPP, we apply Algorithm 1 in Appendix \ref{subsec:appA5} based on the analytic error of spline methods. 
PPPs with different polynomial degrees $p=1,2,3$ are discussed, and we apply two-point Hermite interpolations \cite{Agarwal.1991} (the best ones with known optimal constants) for both $p=2$ and $p=3$ in the step of preparing piecewise polynomials by classical computation. The comparison results for different $\delta$ and $a^2=0.5, 0.1, 0.001$ are demonstrated in Fig.~\ref{comp:fig1}. 
\begin{figure}
\centering
\resizebox{15cm}{!}{
\includegraphics[keepaspectratio]{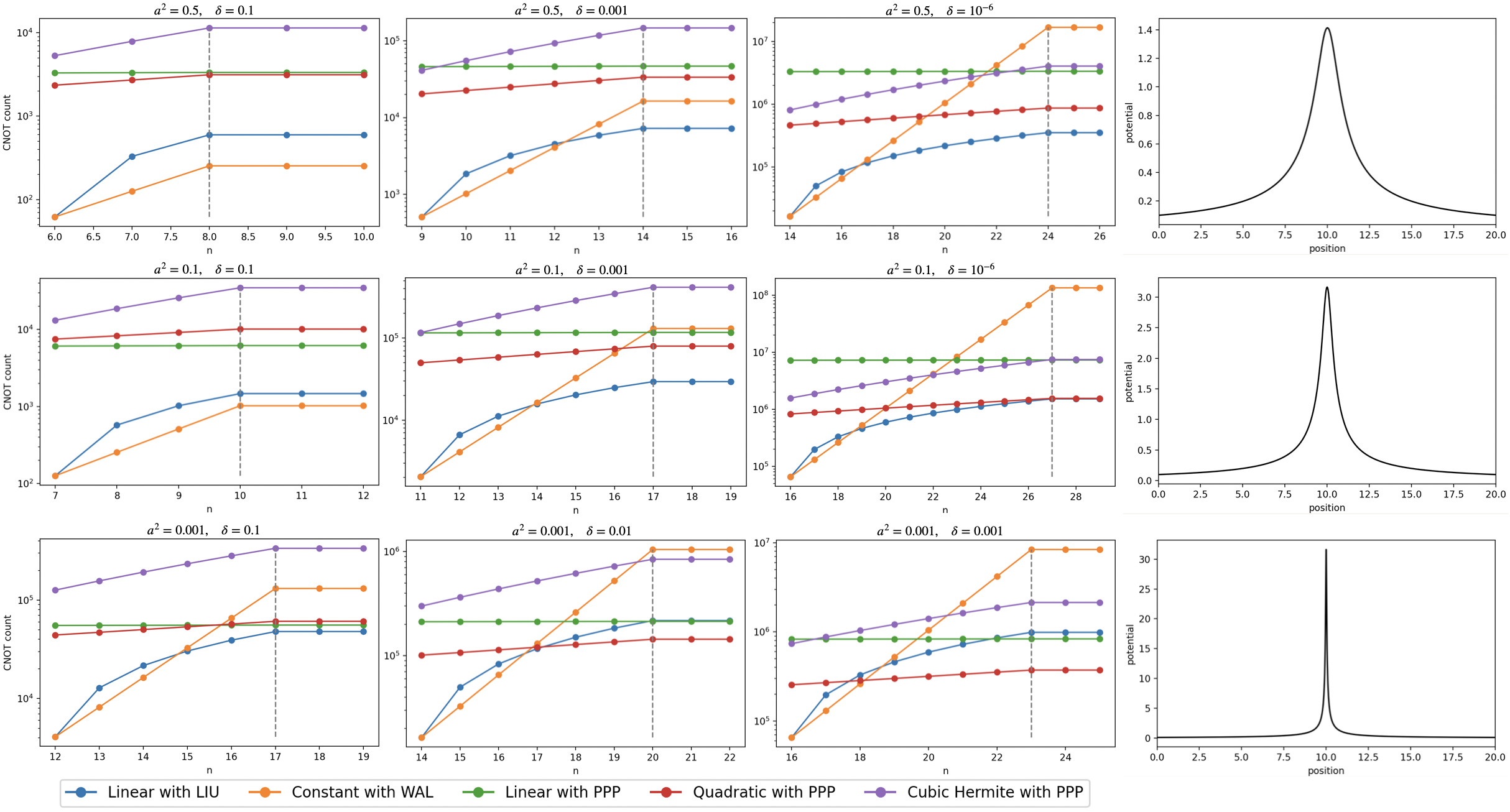}
}
\caption{CNOT count vs. grid parameter $n$ for modified Coulomb potential with several precision $\delta$ and parameter $a^2=0.5,0.1,0.001$ by different methods. 
The vertical dotted line denotes the parameter $m_0$ defined by Eq.~\eqref{comp:m0}, which is the minimal division parameter $m$ such that difference of $V$ at any two adjacent knots of the $2^{m}$ uniform division is smaller than $\delta$. The choices of $n$ are different for each subplot, and we choose $n$ in $\{m_1,m_1+1,\ldots, m_0+2\}$, where $m_0,m_1$ are defined by Eqs.~\eqref{comp:m0}, \eqref{comp:m1}. The potential functions with different $a^2$ are given on the right side of each row for reference. }
\label{comp:fig1}
\end{figure}
We note that for $n=m_0$, we have
$$
\max_{j=0,1,\ldots,N-1}|V(x_{j+1})-V(x_{j})|= \max_{j=0,1,\ldots,N-1}\left|\int_{x_j}^{x_{j+1}}V^{(1)}(x)dx\right|
\le L/2^n \|V^{(1)}\|_\infty = L/2^{m_0} \|V^{(1)}\|_\infty \le \delta.
$$
Then, for $n>m_0$, we only need to apply the approximate circuit to the top $m_0$ qubits to guarantee the desired precision $\delta$. Hence, it is sufficient to consider the range of $n\le m_0$. According to Fig.~\ref{comp:fig1}, we have the following observation:

(i) 
We find WAL/LIU is the best one for small $n\le m_1$ (LIU is reduced to WAL for $n\le m_1$ by its construction). Moreover, for sufficiently large $m_0$, LIU shows better performance than WAL after some $m^\ast\le m_0$. 

(ii) Linear interpolation with LIU is usually better than that with PPP ($p=1$), but for the severe ``singularity" (large derivatives) case, PPP with $p=1$ can perform better than LIU. 

(iii) Although the asymptotic gate count with respect to $\delta$ shows the superiority for large polynomial degree $p\ge 2$, it is not correct for fixed $\delta$ in practice due to the factor $n^p$ that is exponentially increasing in $p$. 
Here, we find PPP with $p=3$ does not outperform either PPP with $p=2$ or LIU in any subplot. Moreover,  PPP with $p=2$ (more precisely, with quadratic Hermite spline) exhibits the best CNOT count for severe ``singularity" and high-precision case when $n$ is sufficiently large (see the subplots of row 3 and columns 2 and 3). 

\begin{example}
\label{exa:2}
$V(x) = Ae^{-ax^2} \cos{\omega x} $, \quad $x\in [0,L]$. 
\end{example}

The second case is a damped oscillating function with parameters $A$, $a$, and $\omega$, where $A$ is the scaling factor of amplitude, $a$ describes the decreasing rate, and $\omega$ denotes the frequency of the sinusoidal function. 
Similarly, we choose $A=1$, $L=20$, $\omega=2$ and plot the CNOT count of the implementations by WAL, LIU, and PPP for a couple of precision $\delta$ and parameter $a$. 
The comparison results for different $\delta$ and $a=1, 0.1, 0.001$ are demonstrated in Fig.~\ref{comp:fig2}. 
\begin{figure}
\centering
\resizebox{15cm}{!}{
\includegraphics[keepaspectratio]{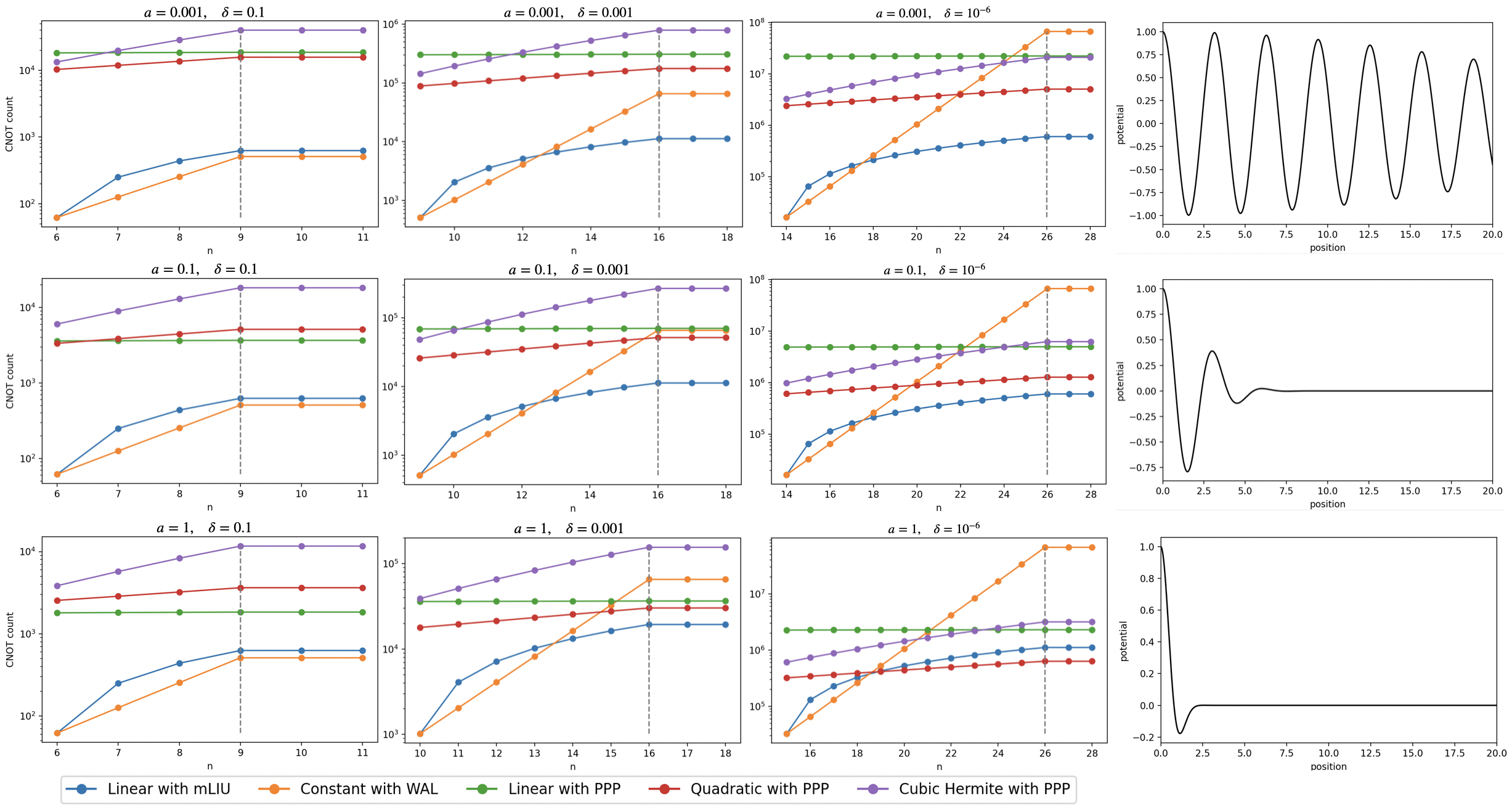}
}
\caption{CNOT count vs. grid parameter $n$ with the damped oscillating function with several precision $\delta$ and parameter $a=1,0.1,0.001$ by different methods. 
The vertical dotted line denotes the parameter $m_0$ defined by Eq.~\eqref{comp:m0}, which is the minimal division parameter $m$ such that difference of $V$ at any two adjacent knots of the $2^{m}$ uniform division is smaller than $\delta$. The choices of $n$ are different for each subplot, and we choose $n$ in $\{m_1,m_1+1,\ldots, m_0+2\}$, where $m_0,m_1$ are defined by Eqs.~\eqref{comp:m0}, \eqref{comp:m1}. The potential functions with different $a$ are given on the right side of each row for reference. }
\label{comp:fig2}
\end{figure}
Again, we provide a plot of CNOT count with respect to $n=m_1,m_1+1,\ldots, m_0+2$ for several choices of $a$ and $\delta$. The CNOT count in the case of $n\ge m_0$ is regarded as constant for each subplot since it is enough to apply the approximate circuit to the top $m_0$ qubits. 
Here, we use mLIU for the reason that $V$ is not periodic. According to Fig.~\ref{comp:fig2}, we have a similar observation that regardless of the oscillation of the function, WAL/mLIU is suitable for the case of low precision and small derivatives, and PPP $(p=2)$ gives the least CNOT count for high-precision and large derivatives case. 

We explain the interesting observation that PPP with $p=2$ outperforms PPP with $p=3$ in practical cases. 
For high-precision case, since we consider $n=\{m_1,m_1+1,\ldots, m_0\}$, which implies $n=\mathcal{O}(\log_2 \delta^{-1})$. Then, the dominant gate count for PPP with $p=2$ is $(\log_2 \delta^{-1})^2\delta^{-1/3}$, while it is $(\log_2 \delta^{-1})^3\delta^{-1/4}$ for PPP with $p=3$. 
By noting that $\delta^{-1/12}$ is smaller than $\log_2\delta^{-1}$ for $10^{-22}\le \delta\le 0.1$, we can verify that PPP with $p=2$ outperforms PPP with a higher degree for most practical cases, e.g., $\delta=10^{-3}$. More precisely, the smallest $\delta$ such that PPP with $p=2$ outperforms other PPPs depends on the given function, and we postpone a detailed analysis to Appendix \ref{subsec:appB4}.  
To find the best method in the sense of minimizing CNOT count, we suggest the following steps. First, we calculate the parameters $m_0,m_1,m_2,\tilde M_1$, and $\tilde M_2$ by the function $V$ and precision $\delta$. Next, we plot the CNOT count of WAL, LIU/mLIU (LIU if periodic and mLIU otherwise), and PPP ($p=2$), respectively, in terms of the analytic estimations (see Appendix \ref{sec:appA} or a summary in Appendix \ref{subsec:appB4}). Finally, we pick up the one with the least CNOT count according to the grid parameter $n$. 
Owing to our ``sharp" analytic estimations of the parameters, the above calculations are cheap on a classical computer, and hence, it is possible to choose a suitable method prior to the construction of the quantum circuits.  

In Figs.~\ref{comp:fig1} and \ref{comp:fig2}, we provide a detailed comparison only for the CNOT count. Since the implementations of WAL, LIU, and PPP are based on CNOT gates and phase gates, roughly speaking, the circuit depth is nearly proportional to the CNOT count. Then, we have a similar discussion for the depth (see Appendix \ref{subsec:appB3} for more details).

\section{Application to first-quantized Hamiltonian simulation}

In this section, we apply the methods to Hamiltonian simulation and estimate the required resources regarding several parameters. We mainly focus on CNOT count since it is an important factor in describing the entanglement of the qubits, and it provides us with information on gate count and circuit depth according to our previous discussion. 

As for the problem of first-quantized Hamiltonian simulation, we employ the well-known grid-based method using the so-called centered/shifted quantum Fourier transform $U_{\text{CQFT}}$ to diagonalize the kinetic energy operator $\hat{T}$. By noting that the potential energy operator $\hat{V}$ itself is diagonal in the real-space representation, we combine this with the Trotter-Suzuki formula to obtain an approximation scheme \cite{Kassal.2008, Kosugi.2023, Childs.2022, Chan.2023}. 
For instance, with the first-order Trotter-Suzuki formula, we have the following approximation
$$
e^{-\mathrm{i}\mathcal{H}K\Delta t} = (e^{-\mathrm{i}\mathcal{H}\Delta t})^K \approx (e^{-\mathrm{i}\hat{T}\Delta t} e^{-\mathrm{i}\hat{V}\Delta t})^K = (U_{\text{CQFT}} U_{\text{kin}}(\Delta t)U_{\text{CQFT}}^{\dag} U_{\text{pot}}(\Delta t))^K,
$$
where $U_{\text{kin}}(\Delta t):=U_{\text{CQFT}}^{\dag}e^{-\mathrm{i}\hat{T}\Delta t} U_{\text{CQFT}}$, $U_{\text{pot}}(\Delta t):=e^{-\mathrm{i}\hat{V}\Delta t}$ are two diagonal unitary matrices with a time step parameter $\Delta t$. 

\subsection{General formulation of a molecular system}
\label{subsec:4-1}

We adopt the formulation introduced in \cite{Kosugi.2023} and consider a molecular system consisting of $N_{\text{e}}$ electrons as quantum mechanical particles and $N_{\text{nuc}}$ nuclei as classical point charges fixed at some known positions $\mathbf{R}_v\in \mathbb{R}^d$ ($v=0,\ldots, N_{\text{nuc}}-1$) where $d\in \mathbb{N}$ is the space dimension. 
We assume these two kinds of particles interact with each other under pairwise interactions, which depend only on the distance between two particles. 
In other words, we consider the Hamiltonian given by 
\begin{align*}
\mathcal{H}(\{\bm{R}_v\}_v) &= \hat T + \hat V(\{\bm{R}_v\}_v) = \hat T + \hat V_{\text{ee}} + \hat V_{\text{en}}(\{\bm{R}_v\}_v) + E_{\text{nn}}(\{\bm{R}_v\}_v) + \hat V_{\text{ext}} \\
&= \sum_{\ell=0}^{N_{\text{e}}-1} \frac{\bm{\hat p}_{\ell}^2}{2m_{\text{e}}} + \frac12 \sum_{\ell,\ell^\prime=0, \ell\not=\ell^\prime}^{N_{\text{e}}-1} v_{\text{ee}}(|\bm{\hat r}_\ell-\bm{\hat r}_{\ell^\prime}|) 
+ \sum_{\ell=0}^{N_{\text{e}}-1}\sum_{v=0}^{N_{\text{nuc}}-1} -Z_v v_{\text{en}}(|\bm{\hat r}_\ell-\bm{R}_v|) \\
&\quad + \frac12 \sum_{v,v^\prime=0, v\not=v^\prime}^{N_{\text{nuc}}-1} Z_{v}Z_{v^\prime} v_{\text{nn}}(|\bm{R}_v-\bm{R}_{v^\prime}|) + \sum_{\ell=0}^{N_{\text{e}}-1} v_{\text{ext}}(\bm{\hat r}_\ell).
\end{align*}
Here, $\bm{\hat r}_\ell$ and $\bm{\hat p}_\ell$ denote the position and momentum operators of the $\ell$-th electron, respectively, and $m_{\text{e}}$ is the mass set to $1$. Moreover, $Z_v$ is the charge of the $v$-th nucleus, while we have already set the charge of the electrons to $-1$. Although it is possible to consider the nuclei in quantum registers, in this paper, we consider only the electrons in quantum registers to minimize the number of required qubits. Furthermore, we can combine the parts $\hat V_{\text{en}}$, $E_{\text{nn}}$, and $\hat V_{\text{ext}}$ as a potential
\begin{align*}
\hat V_{\text{e}}(\{\bm{R}_v\}_v) := \sum_{\ell=0}^{N_{\text{e}}-1} v_{\ell}(\bm{\hat r}_\ell;\{\bm{R}_v\}_v),
\end{align*}
where $v_{\ell}$ is the potential felt by the $\ell$-th electron except for the electron-electron interactions. 

We encode the $N_{\text{e}}$-electron wave function in real space using $n$ qubits for each dimension per electron, as usual in the first-quantized/grid-based formalism. We collect $d n N_{\text{e}}$ qubits as the electronic register: $\ket{\bm{k}_0}_{dn} \otimes \cdots \otimes \ket{\bm{k}_{N_{\text{e}}-1}}_{dn}$ where each $\bm{k}_\ell$ describes $d$ integers of $n$-bit corresponding to the position eigenvalue $\sum_{i=1}^d k_{\ell,i}\bm{e}_i \Delta x_i$. Here, $\bm{e}_i$ is the unit vector, $\Delta x_i = L_i/2^n$, $i=1,\ldots,d$, and $L_i$ is the length of the simulation cell in each dimension. 

Under this formulation, the real-time evolution operator with time step $\Delta t$ by the first-/second-order Trotter-Suzuki formula is shown in Fig.~\ref{appl:Fig1}. 
\begin{figure}
\centering
\resizebox{15cm}{!}{
\includegraphics[keepaspectratio]{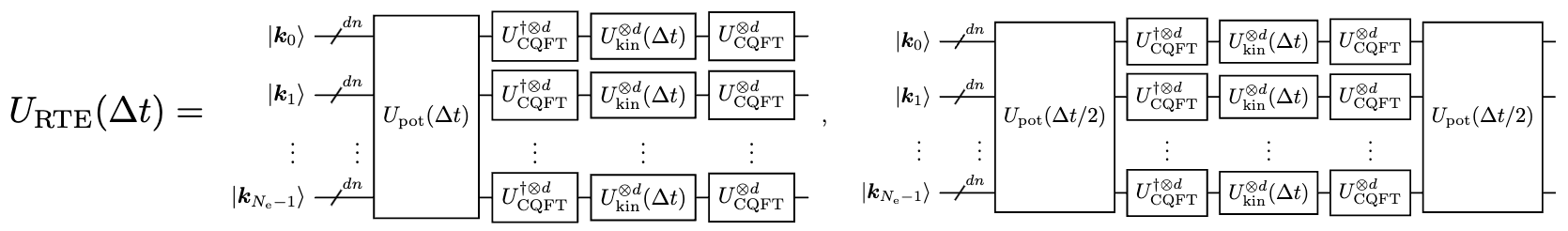}
}
\caption{Implementation of the real-time evolution operator $U_{\text{RTE}}(\Delta t) = e^{-\mathrm{i}\mathcal{H}\Delta t}$ by the Trotter-Suzuki formula. The left circuit is for the first-order Trotter-Suzuki formula, while the right one is for the second-order Trotter-Suzuki formula. 
For each dimension, we apply $U_{\text{CQFT}}$ and $U_{\text{kin}}(\Delta t)$ on the $n$-qubit registers, which means there are $2d N_{\text{e}}$ operations $U_{\text{CQFT}}$/$U_{\text{CQFT}}^\dag$ and $d N_{\text{e}}$ operations $U_{\text{kin}}(\Delta t)$ in total in the above circuits.} 
\label{appl:Fig1}
\end{figure}
Although there are $N_{\text{e}}(N_{\text{e}}-1)$ electron-electron interactions and $N_{\text{e}}$ single electron potentials in $U_{\text{pot}}(\Delta t)$, we find that $U_{\text{pot}}(\Delta t)$ can be implemented in linear depth of $\mathcal{O}(N_{\text{e}})$ by parallelizing the implementations even without ancillary registers (see Appendix \ref{subsec:appC0}).  
In general, implementation of the single electron potential $e^{-\mathrm{i}v_\ell(\bm{\hat r}_\ell; \{\bm{R}_v\}_v)\Delta t}$ needs a phase gate for a $d$-variable function. If $\hat V_{\text{ext}}\equiv 0$, then the electron-electron/electron-nucleus interaction can be implemented using a distance register as Fig.~\ref{appl:Fig2}. 
\begin{figure}
\centering
\resizebox{15cm}{!}{
\includegraphics[keepaspectratio]{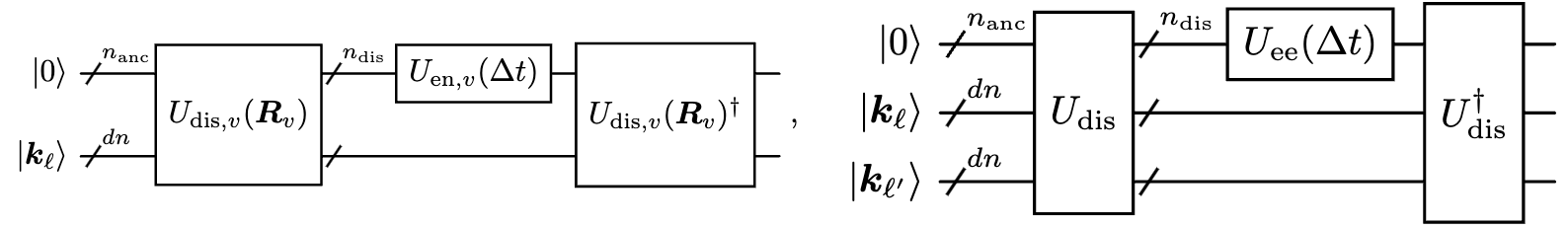}
}
\caption{Implementations of the electron-nucleus interaction $e^{-\mathrm{i}Z_v v_{\text{en}}(|\bm{\hat r}_\ell-\bm{R}_v|)\Delta t}$ for $\ell=0,\ldots,N_{\text{e}}-1$, $v=0,\ldots,N_{\text{nuc}}-1$ and the electron-electron interaction $e^{-\mathrm{i}v_{\text{en}}(|\bm{\hat r}_\ell-\bm{\hat r}_{\ell^\prime}|)\Delta t}$ for $\ell,\ell^\prime=0,\ldots,N_{\text{e}}-1$, $\ell\not=\ell^\prime$. Here, $n_{\text{anc}}$ is the number of ancillary qubits, and $n_{\text{dis}}$ denotes the size of the distance register. Detailed circuits of $U_{\text{dis},v}(\bm{R}_v)$ and $U_{\text{dis}}$ will be given in Appendix \ref{subsec:appC1}.}  
\label{appl:Fig2}
\end{figure}
Here, we put the part of $E_{\text{nn}}$ for the $v$-th nucleus into $U_{\text{en},v}(\Delta t)$, and $U_{\text{en},v}(\Delta t)$ and $U_{\text{ee}}(\Delta t)$ are both diagonal unitary operations that the methods in Sect.~\ref{sec:methods} work. 
As for $U_{\text{dis},v}(\bm{R}_v)$ and $U_{\text{dis}}$, there are usual ways based on arithmetic operations. Besides, we also propose implementations based on QFTs and polynomial phase gates to minimize the number of required qubits. The details are provided in Appendix \ref{subsec:appC1} for completeness. 

\subsection{Resource estimation}

We estimate the required resources for the implementation using the Trotter-Suzuki formula. 
For a given error bound $\varepsilon>0$ and simulation time $t$, we are interested in finding an approximate circuit for $e^{-\mathrm{i}\mathcal{H}t}$ based on the second-order Trotter-Suzuki formula such that 
\begin{align*}
\left\|e^{-\mathrm{i}\mathcal{H}t}-(\tilde U_{\text{pot}}(t/(2K)) U_{\text{CQFT}}^{\otimes d} U_{\text{kin}}^{\otimes d}(t/K)U_{\text{CQFT}}^{\dag\otimes d} \tilde U_{\text{pot}}(t/(2K)))^K\right\| \le \varepsilon.
\end{align*}
In other words, we find $K\in \mathbb{N}$ and some approximate unitary diagonal operation $\tilde U_{\text{pot}}$ to $U_{\text{pot}}$, which depends on the error bound $\varepsilon$. 

As we know, there are three types of errors in the theoretical formulation of the first-quantized Hamiltonian simulation based on the Trotter-Suzuki formula. 

\noindent \underline{Discretization error} 
Roughly speaking, this is an error between the discrete wave function and the continuous wave function subject to the Schr\"odinger equation, which decreases exponentially as the discretization parameter $n$ increases \cite{Childs.2022}. 
Theoretically, one can regard this error as the truncation error of the Fourier series expansion, but this is beyond the interest of this paper. As a result, we assume that the discretization parameter is suitably determined, and we consider only the error between the operation $e^{-\mathrm{i}H_n t}$ and its approximation where $H_n$ is the $2^n \times 2^n$-discretized Hamiltonian matrix of $\mathcal{H}$. That is, we consider
\begin{align}
\label{appl:tareq1}
\left\|e^{-\mathrm{i}H_n t}-(\tilde U_{\text{pot}}(t/(2K)) U_{\text{CQFT}}^{\otimes d} U_{\text{kin}}^{\otimes d}(t/K)U_{\text{CQFT}}^{\dag\otimes d} \tilde U_{\text{pot}}(t/(2K)))^K\right\| \le \varepsilon.
\end{align}

\noindent \underline{Trotter-Suzuki error} 
According to the error estimate for the second-order Trotter-Suzuki formula \cite{Jahnke.2000}, we have
\begin{align*}
&\quad \left\|e^{-\mathrm{i}H_n t}-(U_{\text{pot}}(t/(2K)) U_{\text{CQFT}}^{\otimes d} U_{\text{kin}}^{\otimes d}(t/K)U_{\text{CQFT}}^{\dag\otimes d} U_{\text{pot}}(t/(2K)))^K\right\| \\
&=\left\|e^{-\mathrm{i}(T_n+V_n) t}-\left(e^{-\mathrm{i}V_n t/(2K)} e^{-\mathrm{i}T_n t/(2K)} e^{-\mathrm{i}V_n t/(2K)}\right)^K\right\|
\le CK(t/K)^3,
\end{align*}
where $C$ is some constant depending on the discretized matrices $T_n$ and $V_n$. 
If we assume the right-hand side above is upper bound by $\varepsilon/2$, then we reach an estimation of $K$ as 
$$
K = \mathcal{O}(t^{3/2}\varepsilon^{-1/2}). 
$$

\noindent \underline{Approximation error} 
By noting the unitarity of the operations $e^{-\mathrm{i}T_n t/K}$ and $e^{-\mathrm{i}V_n t/(2K)}$, we estimate
\small
\begin{align*}
&\quad \left\|(\tilde U_{\text{pot}}(t/(2K)) U_{\text{CQFT}}^{\otimes d} U_{\text{kin}}^{\otimes d}(t/K)U_{\text{CQFT}}^{\dag\otimes d} \tilde U_{\text{pot}}(t/(2K)))^K
-(U_{\text{pot}}(t/(2K)) U_{\text{CQFT}}^{\otimes d} U_{\text{kin}}^{\otimes d}(t/K)U_{\text{CQFT}}^{\dag\otimes d} U_{\text{pot}}(t/(2K)))^K\right\| \\
&\le 2K\left\|\tilde U_{\text{pot}}(t/(2K))-U_{\text{pot}}(t/(2K))\right\|
\le 2K t/(2K) \|\tilde V_n - V_n\|_\infty = t \|\tilde V_n - V_n\|_\infty. 
\end{align*}
According to Fig.~\ref{appl:Fig2}, we apply the approximate circuit with precision $\delta$ for each implementation of the electron-nucleus/electron-electron interaction part. 
Taking the additive error into account, we have 
\begin{align*}
t \|\tilde V_n - V_n\|_\infty \le t \delta (N_{\text{e}} N_{\text{nuc}}+N_{\text{e}}(N_{\text{e}}-1)/2),
\end{align*}
which implies $\delta = \varepsilon/(t(2N_{\text{e}} N_{\text{nuc}}+N_{\text{e}}(N_{\text{e}}-1)))$ so that $t\|\tilde V_n - V_n\|_\infty \le \varepsilon/2$, and hence, 
Eq.~\eqref{appl:tareq1} is satisfied. 

The distance register needs $2n+\lceil\log_2 d\rceil$ qubits to describe a value in $[0, dL^2)$, which yields the asymptotic gate count of PPP (see Appendix \ref{subsec:appA3} with $L$ replaced by $dL^2$) for each approximation: 
$$
\mathcal{O}\left(dt^{1/(p+1)}N_{\text{tot}}^{2/(p+1)}\varepsilon^{-1/(p+1)}(n^p(\log d)^p+(\log (dtN_{\text{tot}}^2/\varepsilon)^2))\right) 
= \mathcal{\tilde O}(dn^p t^{1/(p+1)}N_{\text{tot}}^{2/(p+1)}\varepsilon^{-1/(p+1)}), 
$$
concerning the parameters $n, d, t, N_{\text{tot}},\varepsilon$ where $N_{\text{tot}}=N_{\text{e}}+N_{\text{nuc}}$ is the total number of particles. 
The hidden linear dependence on $d$ comes from the number of sub-intervals $\tilde M$ because $\tilde M$ is proportional to $dL^2$ in the worst case. In fact, for a ``localized" potential such that the integral regarding its derivative is independent of $d$, $\tilde M$ is uniform with respect to $d$, and the above linear dependency will vanish. 
Moreover, the circuit depth is in the same order as the gate count due to Table \ref{tab:comp2}. 
As for the second-order Trotter-Suzuki formula, we need $K+1$ times implementations of $\tilde U_{\text{pot}}$ and $KdN_{\text{e}}$ times implementations of $U_{\text{kin}}$ (an operation on $n$ qubits). 
Therefore, the total gate count of the Hamiltonian simulation is $\mathcal{\tilde O}\left(d N_{\text{tot}}^{2+2/(p+1)}t^{3/2+1/(p+1)}\varepsilon^{-1/2-1/(p+1)}\right)$, and its circuit depth is $\mathcal{\tilde O}\left(d N_{\text{tot}}^{1+2/(p+1)}t^{3/2+1/(p+1)}\varepsilon^{-1/2-1/(p+1)}\right)$. 
Here, 
the dependence on $N_{\text{tot}}$ for the depth is one order smaller than that for the gate count because we use parallel implementations of the potential part $\tilde U_{\text{pot}}(t/(2K))$ we mentioned in the last paragraph of Sect.~\ref{subsec:4-1}, and this saves the depth by a factor of $N_{\text{e}}$. However, we need $(2n+1+\lceil\log_2 d\rceil)$ ancillary qubits for each electron in this case. Based on PPP and QFT-based/arithmetic-operation-based implementations of the distance registers, we summarize the required resources for the Hamiltonian simulation ($d\ge 2$) in Table \ref{tab:appl1}. We have omitted the polynomial dependence $\mathcal{O}(n^p)$ in gate count/circuit depth because we apply the approximate circuit only to the top $m_0$ qubits of the circuit if $n>m_0$. According to the definition of $m_0$ in Eq.~\eqref{comp:m0}, this yields a negligible factor in gate count/circuit depth of $\mathcal{O}((\log \varepsilon^{-1})^{p})$. 
\begin{table}[htb]
\centering
\caption{Resource estimation for Hamiltonian simulation with second-order Trotter-Suzuki formula. }
\label{tab:appl1}
\scalebox{0.58}[0.58]{
\begin{tabular}{l|cccc}
\hline
& Gate count & Depth & Ancilla count & Qubit count\\
\hline
QFT (sequential) & $\mathcal{\tilde O}\left(d N_{\text{tot}}^{2+2/(p+1)}t^{3/2+1/(p+1)}\varepsilon^{-1/2-1/(p+1)}\right)$ 
& $\mathcal{\tilde O}\left(d N_{\text{tot}}^{2+2/(p+1)}t^{3/2+1/(p+1)}\varepsilon^{-1/2-1/(p+1)}\right)$ 
& $2n+1+\lceil\log_2 d\rceil$ 
& $(dN_{\text{e}}+2)n+1+\lceil\log_2 d\rceil$\\
Arithm. (sequential) & $\mathcal{\tilde O}\left(d N_{\text{tot}}^{2+2/(p+1)}t^{3/2+1/(p+1)}\varepsilon^{-1/2-1/(p+1)}\right)$ 
& $\mathcal{\tilde O}\left(d N_{\text{tot}}^{2+2/(p+1)}t^{3/2+1/(p+1)}\varepsilon^{-1/2-1/(p+1)}\right)$ 
& $2dn+d(d+1)/2$ 
& $(N_{\text{e}}+2)dn+d(d+1)/2$\\
QFT (parallel) & $\mathcal{\tilde O}\left(d N_{\text{tot}}^{2+2/(p+1)}t^{3/2+1/(p+1)}\varepsilon^{-1/2-1/(p+1)}\right)$ 
& $\mathcal{\tilde O}\left(d N_{\text{tot}}^{1+2/(p+1)}t^{3/2+1/(p+1)}\varepsilon^{-1/2-1/(p+1)}\right)$ 
& $N_{\text{e}}(2n+1+\lceil\log_2 d\rceil)$ 
& $N_{\text{e}}((d+2)n+1+\lceil\log_2 d\rceil)$\\
Arithm. (parallel) & $\mathcal{\tilde O}\left(d N_{\text{tot}}^{2+2/(p+1)}t^{3/2+1/(p+1)}\varepsilon^{-1/2-1/(p+1)}\right)$ 
& $\mathcal{\tilde O}\left(d N_{\text{tot}}^{1+2/(p+1)}t^{3/2+1/(p+1)}\varepsilon^{-1/2-1/(p+1)}\right)$ 
& $N_{\text{e}}(2dn+d(d+1)/2)$ 
& $N_{\text{e}}(3dn+d(d+1)/2)$\\
\hline
\end{tabular}
}
\end{table}

To clarify, we compare and emphasize the difference between our estimation and the result in \cite[Theorem 6]{Childs.2022}. 
Childs et al. \cite{Childs.2022} discussed the gate complexity considering the discretization and the Trotter-Suzuki errors by neglecting the approximate error in implementing the potential part.  
In comparison, as we mentioned, our main target in this paper is how the approximation error contributes to the required resources. Thus, we focus on the Trotter-Suzuki and the approximation errors in our estimation by assuming that $n$ is suitably determined. 
Using the second-order Trotter-Suzuki formula, \cite[Theorem 6]{Childs.2022} gives the asymptotic gate count $\mathcal{\tilde O}((dN_{\text{tot}})^{5/2}t^{3/2}\varepsilon^{-1/2})$ in our notation, where the worse $(dN_{\text{tot}})$-dependence owes to the discretization error. 
Then, our result demonstrates an additional dependence of $\mathcal{\tilde O}((t/\varepsilon)^{1/(p+1)})$ in gate count/circuit depth if the approximation error is addressed (using piecewise $p$-degree polynomial approximation). As we discussed in Sect.~\ref{sec:comp}, $p$ usually takes a small integer, e.g., $p=2$ in practical cases. Then, our result shows a comparable polynomial dependence in gate count/depth coming from the approximation error as that coming from the Trotter-Suzuki error, which is a novel finding that one has to consider the approximation error in addition to the Trotter-Suzuki error. 

We end this subsection with a crucial remark on the choice of discretization parameter $n$. Although one can theoretically estimate the discretization error \cite{Childs.2022}, it is not helpful in practice since the result gives only the order, and it seems impossible to derive a sharp number of $n$ even for a detailed Hamiltonian. Thus, the discretization parameter $n$ will probably be chosen by testing or by experience in practice. Since our resource estimation is based on PPP, it is effective when we assume small error $\varepsilon$ and large $n$. In contrast, if $\varepsilon$ and $n$ lie in the region where the exact implementation (WAL with $m_0=n$) is selected, then we have no approximation error and the total gate count of the Hamiltonian simulation is $\mathcal{\tilde O}(d N_{\text{tot}}^2 t^{3/2}\varepsilon^{-1/2})$. 

\section{Conclusion}

In this paper, we focus on a systematic study of the approximate implementation of the real-time evolution operator for potentials with at most one ancillary qubit for a given error bound, which is an important task for first-quantized Hamiltonian simulation on a near-term quantum computer. 

First, we summarize the previous methods WAL and APK, and propose an original method LIU and a modified method PPP. 
Using the explicit implementation of each method, we derive the asymptotic gate count and circuit depth regarding discretization parameter $n$ and precision $\delta$ and find that PPP with a high-degree polynomial outperforms WAL/LIU in the order of precision. Moreover, APK has an even better performance than PPP in the scaling of the pre-factor at the cost of using $\mathcal{O}(\log 1/\delta)$ ancillary qubits. 

Second, we conducted a case study with WAL, LIU, and PPP for the modified Coulomb potential and a damped oscillating function. By employing the theoretical optimal constants in the error estimation for spline interpolation, we divided the interval of $n$ into three sub-intervals $I_1=[1,m_1)$, $I_2=[m_1,m_0]$, and $I_3=(m_0,\infty)$. Here, $m_0$ and $m_1$ are integers calculated from $\delta$, $V$, and the optimal constants. In a practical setting, our results show that for $n\in I_1$, LIU is equivalent to WAL and WAL/LIU is better than PPP. Moreover, for $n\in I_3$, it is sufficient to apply the approximate circuit to the top $m_0$ qubits. 
Hence, we only need to compare the gate count/depth for different methods with $n=m_0$. 
Furthermore, for $n\in I_2$, we determine whether PPP outperforms WAL and LIU depending on $\delta$, $n$, and the norms of the derivatives of the potential. If the norms are sufficiently large and $\delta$ is sufficiently small, then there is some integer $n^\ast\in I_2$ such that PPP is better than WAL and LIU for $n\ge n^\ast$. Among the PPPs with different polynomial degrees, it appears that $p=2$ (quadratic Hermite interpolation) works best in most practical settings of $\delta$. 

Third, we adopted the formulation for a molecular system in \cite{Kosugi.2023}, and by applying the grid-based method, we estimated the gate count and circuit depth for the first-quantized Hamiltonian simulation. Because we are keeping in mind the high-precision case, resources estimation based on PPP is addressed. In addition to the gate count of $\mathcal{\tilde O}(t^{3/2}\varepsilon^{-1/2})$ concerning the second-order Trotter-Suzuki error, we show an extra overhead of $\mathcal{\tilde O}(t^{1/(p+1)}\varepsilon^{-1/(p+1)})$ resulting from the approximation error, where $t$ is the simulation time, $\varepsilon$ is the error bound, and $p$ is the polynomial degree. This implies that the approximation error is not negligible for long-time simulations under the requirement of a small error bound. 

Finally, we provided comprehensive references and novel discussions on the technical details of the implementations, as well as their gate counts, in the appendix. In addition, we proposed two closely related algorithms for finding the coefficients of polynomials on a classical computer, which helped us develop a complete numerical scheme for deriving an approximate quantum circuit of the real-time evolution operator from the input potential $V$, discretization parameter $n$, polynomial degree $p$, and precision $\delta$. 

As an extension of this study, it may be interesting to consider a case in which the potential is given by a general multivariable function. For the one-variable (distance) potential discussed in this paper, distance registers must be implemented, as shown in Fig.~\ref{appl:Fig2}. In particular, this yields an additional cost and introduces a linear dependence on $N_{\text{nuc}}$ in the gate count and circuit depth because the distance registers are derived one by one. In a general setting of multivariable potential, one can implement the electron-nucleus interaction part and external potential part simultaneously, which seems to be more efficient. 

In this paper, we focus on methods with minimal circuit width (number of qubits). As we found in the first part of Sect.~\ref{subsec:comp-1}, APK is a promising method for a high-precision case in the sense that its asymptotic gate count and circuit depth would be even smaller at the cost of using many ancillary qubits. Therefore, there seems to be a trade-off between circuit depth and width. A detailed discussion in this direction is indispensable for practical applications because of the limitation on the maximal depth/width of circuits that can be processed on near-term quantum computers. 

\section*{Acknowledgment}
This work was supported by Japan Society for the Promotion of Science (JSPS) KAKENHI under Grant-in-Aid for Scientific Research No.21H04553, No.20H00340, and No.22H01517. 
This work was partially supported by the Center of Innovations for Sustainable Quantum AI (JST Grant number JPMJPF2221). 

\appendix

\section{Details on the methods}
\label{sec:appA}

\subsection{Sequency-ordered Walsh operator}
\label{subsec:appA1}

In Zhang et al. \cite{Zhang.2022}, the authors optimized the previous circuit in depth by parallelizing quantum gates according to a constructive algorithm. For completeness, we give the circuit for $n=4$ as an example in Fig.~\ref{app:Fig1}. 
\begin{figure}
\centering
\resizebox{15cm}{!}{
\includegraphics[keepaspectratio]{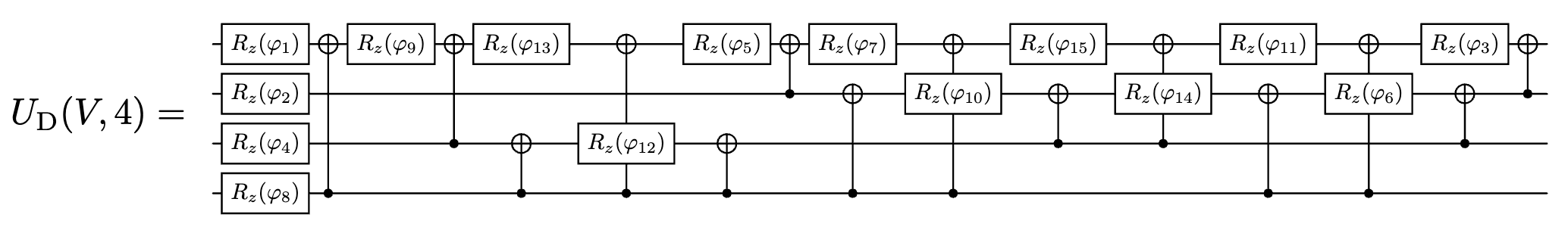}
}
\caption{Example of quantum circuit by sequency-ordered Walsh operators with $n=4$ \cite{Zhang.2022}. Here, $\varphi_j$ is the rotation angle derived by applying the Hadamard matrix to the diagonal vector of the discretized potential matrix. }
\label{app:Fig1}
\end{figure}
For a $2^n\times 2^n$ unitary diagonal matrix, such an implementation uses a global phase gate, $2^n-1$ $R_z$ gates, and $2^n-2$ CNOT gates with depth $2^n$. 
On the other hand, Welch et al. suggested an approximate circuit by applying the Walsh operators only to the top $m<n$ qubits \cite{Welch.2014}. 
Since this is equivalent to conducting a piecewise constant approximation in a $2^m$ uniform grid, we can estimate $m$ for a given precision $\delta$ as
\begin{align*}
\max_{x\in [0,L]}|V^{(1)}(x)|\frac{L}{2^m} \le \delta,
\end{align*}
which implies that the difference of $V$ at arbitrary two adjacent grid points is upper bound by $\delta$. Thus, we can take  
\begin{align*}
m = \left\lceil \log_2 \left(L\frac{\|V^{(1)}\|_\infty}{\delta}\right) \right\rceil.
\end{align*}
Here and henceforth, $V^{(k)}$ is the $k$-th derivative of $V$, and $\|\cdot\|_\infty$ denotes the maximum norm.
In this sense, the approximate circuit by \cite{Welch.2014} uses a global phase gate, $2^m-1$ $R_z$ gates, and $2^m-2$ CNOT gates. 
Moreover, the depth based on the gate set $\{\mathrm{CNOT}, R_z\}$ is $2^m$. 
\cite{Welch.2014} also mentioned the possibility of the reduction by truncating the Walsh operators with small coefficients, but this would introduce an additional truncation error that is hard to estimate analytically.  

\subsection{Linear interpolated unitary diagonal matrix}
\label{subsec:appA2}

According to Fig.~\ref{met:Fig2}, LIU uses $2(n-m)+1$ queries to the $\mathbb{C}^{M\times M}$ unitary diagonal operation and $2(n-m)$ controlled increment/decrement gates on $m$ qubits. 
The former can be realized by $2^m-1$ $R_z$ gates and $2^m-2$ CNOT gates with depth $2^m$ for each query \cite{Zhang.2022}, whereas the latter can be directly implemented by a couple of multi-controlled NOT gates by noting Fig.~\ref{app:Fig2}. 
\begin{figure}
\centering
\resizebox{13cm}{!}{
\includegraphics[keepaspectratio]{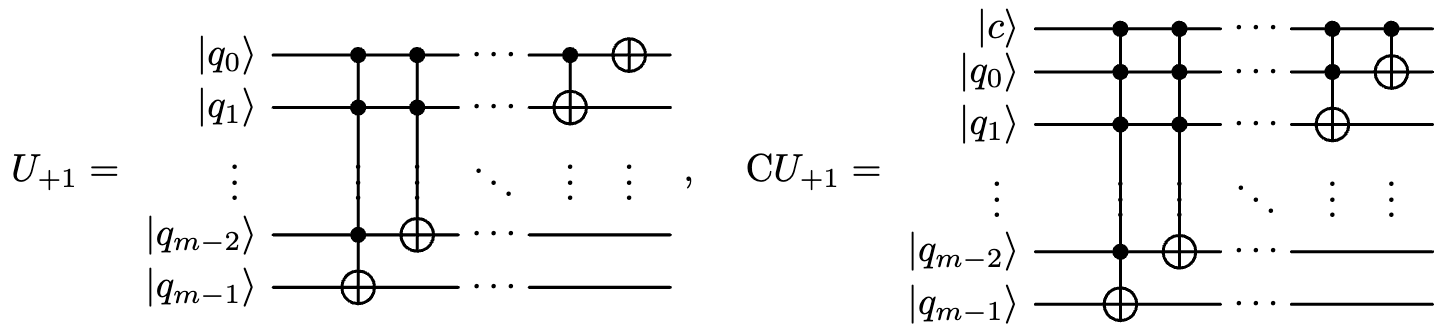}
}
\caption{Quantum circuit of increment gate $U_{\text{+1}}$ and its controlled version according to Exercise 4.27 in \cite{NC2010}. }
\label{app:Fig2}
\end{figure}
Then, LIU needs a global phase gate, $(2^m-1)(2(n-m)+1)$ $R_z$ gates, $(2^m-2)(2(n-m)+1)+2(n-m)$ CNOT gates, and $2(n-m)$ $j$-controlled NOT gates for each $j=2,3,\ldots,m$. 
Using several ancillary qubits, $j$-controlled NOT gates can be realized by CNOT gates and single-qubit gates of linear order $\mathcal{O}(j)$ \cite{Barenco.1995, Gidney2015}. 
Because we try to avoid using ancillas in this paper, we mention another way to implement $U_{\text{+1}}$ based on a quantum Fourier transform (QFT), its inverse, and phase gates, which one can find in \cite[Fig. 2]{Yuan.2023}. 

If one applies the conventional way for the QFT and its inverse, then one needs a global phase gate, $2m$ Hadamard gates, $m$ phase gates, and $m(m-1)$ controlled phase gates for a single increment gate. 
Moreover, a controlled increment gate can be constructed by a global phase gate, $2m$ Hadamard gates, and $m^2$ controlled phase gates since we only need to apply the control to the part of phase gates. 
Under this implementation, LIU alternatively uses a global phase gate, $4m(n-m)$ Hadamard gates, $(2^m-1)(2(n-m)+1)$ phase gates, $(2^m-2)(2(n-m)+1)$ CNOT gates, and $2m^2(n-m)$ controlled phase gates. 
Based on the gate set $\{\mathrm{H}, \mathrm{CNOT}, R_z\}$, the circuit depth is bound by $2^m(2(n-m)+1)+2(n-m)(16m-16)$ because the depth of a QFT circuit or its inverse is $8m-10$, and the controlled phase gates between the QFT and its inverse can be partly implemented with the QFT in parallel, which introduces only an additional depth of 4. 


Next, we determine the parameter $m$ by a given precision $\delta>0$. 
In terms of a classical error estimate for linear interpolation, e.g., \cite{Agarwal.1991, Hall.1976}, we have
\begin{align*}
\max_{j=0,1,\ldots,N-1}|v_j-\tilde v_j| \le \frac{1}{8}\max_{x\in [0,L]}|V^{(2)}(x)| \left(\frac{L}{M}\right)^2 \le \delta,
\end{align*}
where $V\in C^2[0,L]$ is a sufficiently smooth function. Here, $\tilde v_j$ takes the value of the linear interpolation function $\tilde V$ at each uniform knot $x_j = jL/N$, $j=0,1,\ldots,N$.  
Therefore, we have the estimation of $m$ by 
\begin{equation}
m = \left\lceil \log_2 \left(L\sqrt{\frac{\|V^{(2)}\|_\infty}{8\delta}}\right) \right\rceil.
\end{equation}
If the above integer is larger than $n$, then we simply choose $m=n$. In other words, we do not execute linear interpolation, and LIU is equivalent to WAL.  

\subsection{Phase gate for piecewise-defined polynomial}
\label{subsec:appA3}

Similarly to \cite{Kosugi.2023}, we provide an efficient circuit based on multi-controlled phase gates for a single polynomial phase gate: 
\begin{align*}
U_{\text{ph}}[f] = \sum_{j=0}^{N-1} \mathrm{exp}\left(-\mathrm{i}f(x_j)\right) \ket{j}\bra{j}.
\end{align*} 
Here, we assume that $f$ is a $p$-th degree real-valued polynomial function satisfying
$$
f(x) = \sum_{k=0}^{p} a_k x^k, \quad x\in [0,L],
$$
for $a_k\in \mathbb{R}$, $k=0,1,\ldots,p$. Recall that $x_j = jL/N$, and $j\in \{0,1,\ldots,N-1\}$ admits the binary representation $[j_{n-1}j_{n-2}\cdots j_0]$. 
Then, we calculate
\begin{align*}
a_k x_j^k &= a_k j^k (L/N)^k \\
&= a_k (L/N)^k \left(\sum_{\ell=0}^{n-1}j_\ell 2^{\ell}\right)^k \\
&= \sum_{\ell_0=0}^{n-1}\cdots \sum_{\ell_{k-1}=0}^{n-1} a_k (L/N)^k 
 2^{\ell_0+\cdots +\ell_{k-1}} \delta_{j_{\ell_0},1} \cdots \delta_{j_{\ell_{k-1}},1}.   
\end{align*}
Thus, we have
\begin{align*}
U_{\text{ph}}[f] \ket{j} &= \mathrm{exp}\left(-\mathrm{i}f(x_j)\right) \ket{j} \\
&= \prod_{k=0}^p \mathrm{exp}\left(-\mathrm{i}a_k x_j^k\right) \ket{j} \\
&= \prod_{k=0}^p \left(\prod_{\ell_0=0}^{n-1}\cdots \prod_{\ell_{k-1}=0}^{n-1} \mathrm{exp}\left(-\mathrm{i}a_k (L/N)^k 2^{\ell_0+\cdots +\ell_{k-1}} \delta_{j_{\ell_0},1} \cdots \delta_{j_{\ell_{k-1}},1}\right)\right) \ket{j} \\
&= \prod_{k=0}^p \left(\prod_{\ell_0=0}^{n-1}\cdots \prod_{\ell_{k-1}=0}^{n-1} \mathrm{C}Z_{\ell_0,\ldots,\ell_{k-1}}\left(-a_k (L/N)^k 2^{\ell_0+\cdots +\ell_{k-1}}\right)\right) \ket{j},
\end{align*}
where $\mathrm{C}Z_{\ell_0,\ldots,\ell_{k-1}}\left(\theta\right)$ denotes a (multi-)controlled phase gate, which yields a phase factor $e^{\mathrm{i}\theta}$ if and only if the $\ell_0$-th, \ldots, $\ell_{k-1}$-th qubits are all in state $\ket{1}$ (repeated indices are allowed). Equivalently, it conducts $Z(\theta) = \ket{0}\bra{0}+e^{\mathrm{i}\theta}\ket{1}\bra{1}$ to an arbitrary qubit among $\ket{q_{\ell_0}},\ldots,\ket{q_{\ell_{k-1}}}$ with the others playing the role of control qubits. 
Therefore, the implementation of $U_{\text{ph}}[f]$ uses at most $\sum_{k=0}^p n^k = (n^{p+1}-1)/(n-1)$ (multi-)controlled phase gates. 
More efficiently, the (multi-)controlled phase gates with the same target and control qubits can be combined into a single (multi-)controlled phase gate by adding their rotation angles, which implies that we need only $\sum_{k=0}^{\min\{p,n\}} \mathrm{C}_k^n$ multi-controlled phase gates (i.e., $\mathrm{C}_k^n$ $(k-1)$-controlled phase gates for $k=1,2,\ldots,\min\{p,n\}$ with a global phase gate) instead. Here, $\mathrm{C}_k^n$ denotes the number of k-combinations of an $n$-element set. 
For $p<n$, the total number of (multi-)controlled phase gates is both $\mathcal{O}(n^p)$, while for $p\ge n$, the total number of (multi-)controlled phase gates is reduced from $\Omega(n^n)$ to $2^n$. 
If we adopt the techniques for a generic multi-controlled single-qubit gate proposed by Silva and Park \cite{SP2022}, the depth of a $j$-controlled phase gate is $\mathcal{O}(j)$ even without any ancillary qubit. As a result, the depth of the total circuit for a $p$-th degree polynomial phase gate $U_{\text{ph}}[f]$ is $\mathcal{O}(pn^p)$. 

Now assume the piecewise polynomial approximation 
\begin{equation*}
\tilde V(x) = \sum_{\ell=0}^{\tilde M-1} \chi_{[\tilde x_{\ell}, \tilde x_{\ell+1})}(x) f_{\ell}(x), \quad x\in [0,L],
\end{equation*}
where $f_\ell$ is a polynomial function with degree $p_{\ell}$ for each $\ell=0,1,\ldots,\tilde M-1$. 
According to the above discussion, the phase gate for a $p$-th degree polynomial is implemented by a global phase gate, $\mathrm{C}_1^n$ phase gates, $\mathrm{C}_2^n$ controlled phase gates, $\mathrm{C}_3^n$ $2$-controlled phase gates, \ldots, $\mathrm{C}_{\min\{p,n\}}^n$ $(\min\{p,n\}-1)$-controlled phase gates. Similarly, the controlled phase gate for a $p$-th degree polynomial consists of a phase gate, $\mathrm{C}_1^n$ controlled phase gates, $\mathrm{C}_2^n$ $2$-controlled phase gates, $\mathrm{C}_3^n$ $3$-controlled phase gates, \ldots, $\mathrm{C}_{\min\{p,n\}}^n$ $\min\{p,n\}$-controlled phase gates. 
On the other hand, the quantum comparator with a given integer in \cite{Yuan.2023} is explicitly built by a QFT and its inverse quantum Fourier transform (IQFT) on $m$ qubits, a QFT and its IQFT on $m+1$ qubits, as well as $2m+1$ phase gates. A conventional implementation of QFT on $m$ qubits requires $m$ Hadamard gates, $m(m-1)/2$ controlled phase gates with a possible global phase gate, and $\lfloor m/2\rfloor$ swap gates. 
Since the swap gates of QFT can be canceled in the circuit in \cite{Yuan.2023} by the swap gates of its inverse, we find that the quantum comparator uses (at most) a global phase gate, $4m+2$ Hadamard gates, $2m+1$ phase gates, and $2m^2$ controlled phase gates. 
As a result, PPP needs a global phase gate, $(8m+4)(\tilde M-1)$ Hadamard gates, $n+(4m+3)(\tilde M-1)$ phase gates, $n(n-1)/2+(4m^2+n)(\tilde M-1)$ controlled phase gates, $\mathrm{C}_3^n+\mathrm{C}_2^n (\tilde M-1)$ $2$-controlled phase gates, \ldots, $\mathrm{C}_{\min\{p,n\}}^n+\mathrm{C}_{\min\{p,n\}-1}^n(\tilde M-1)$ $(\min\{p,n\}-1)$-controlled phase gates, and $\mathrm{C}_{\min\{p,n\}}^n(\tilde M-1)$ $\min\{p,n\}$-controlled phase gates in total. 

In particular, for $p=1$, PPP requires a global phase gate, $(8m+4)(\tilde M-1)$ Hadamard gates, $n+(4m+3)(\tilde M-1)$ phase gates, and $(4m^2+n)(\tilde M-1)$ controlled phase gates. Based on the gate set $\{\mathrm{H}, \mathrm{CNOT}, R_z\}$, the depth for $U_{\text{ph}}$ is $1$, the depth for controlled $U_{\text{ph}}$ is (at most) $3n+1$ (each controlled phase gate consists of $3$ phase gates and $2$ CNOT gates with depth $4$, one of the phase gates for these controlled phase gates can be implemented in parallel with the controlled global phase gate), and the depth for the comparator \cite{Yuan.2023} is $16m+16m-16$ because the depth of a QFT/IQFT on $k$ qubits is $8k-10$, $k=m,m+1$,  and the controlled phase gates between the QFT and IQFT can be partly implemented in parallel with the QFT, which introduces an additional depth of $4$. Therefore, the circuit depth is bound by $(\tilde M-1)(4n+64m-32)+1$ for $p=1$. 
As for $p=2$, PPP requires a global phase gate, $(8m+4)(\tilde M-1)$ Hadamard gates, $n+3n(n-1)/2+(7n(n-1)/2+3n+12m^2+4m+3)(\tilde M-1)$ phase gates, and $n(n-1)+(4n(n-1)+2n+8m^2)(\tilde M-1)$ CNOT gates. The depth for $U_{\text{ph}}$ is (at most) $4n$, the depth for controlled $U_{\text{ph}}$ is (at most) $(6n+2)n$ 
and thus, the circuit depth is upper bound by $(\tilde M-1)(6n^2+2n+64m-32)+4n$. If we consider the parallel implementation for the polynomial phase gate ($n(n-1)/2$ controlled phase gates and $n$ phase gates can be implemented in depth $n$), then we can use the depth-efficient circuit for multiple controlled phase gates with the same control qubit \cite{Gidney2017} to reduce the depth from $\mathcal{O}(\tilde M n^2)$ to $\mathcal{O}(\tilde M n\log_2 n)$. 
For general $p\ge 2$, the gate count and depth depend also on $p$. For example, a naive implementation of a $p$-controlled phase gate without any ancillary qubits uses $\mathcal{O}(p^2)$ single-qubit gates and CNOT gates. We conjecture that the depth in this case is $\mathcal{O}(\tilde M n^p)$. 
More precisely, we mention that the above estimations are based on a common implementation of a multi-controlled phase gate and thus give an overhead on the gate count of PPP. In this paper, the CNOT gates for $2$-controlled and $3$-controlled phase gates are counted as $8$ and $20$, respectively. Although they can be reduced to $6$ and $14$ separately using smarter implementations \cite{Amy.2019}, this does not change the crucial dependence on $n$, $\delta$, etc.

Next, we estimate the parameters $m$ and $\tilde M$ for a given precision $\delta>0$. 
Here, we fix the polynomial degree $p\ge 1$ and give analytic estimations of these parameters. 
Recall that with a coarse-graining parameter $m$, we select the internal knots for polynomial interpolation among $x_{Nk/M} = kL/M$, $k=1,\ldots, M-1$, and the distance between two adjoint knots is at most $L/M$. 
Using the classical error estimation for the spline method \cite{Agarwal.1991, Hall.1976, Dubeau.1996}, assuming $V\in C^{p+1}[0,L]$, we have
\begin{align}
\label{app:err-est}
\max_{j=0,\ldots,N-1}|v_j-\tilde V(x_j)| \le \max_{x\in [0,L]}|V(x)-\tilde V(x)| \le C_p\max_{x\in [0,L]}|V^{(p+1)}(x)| \left(\frac{L}{M}\right)^{p+1} \le \delta.
\end{align}
Here, $C_p$ is a pre-constant depending on degree $p$. The optimal pre-constant is known for several types of spline methods, which can be found in the above references. The word ``optimal" means that the second inequality becomes equality in the worst case (see Appendix \ref{subsec:appB1} for more details). 
Thus, we estimate $m$ for given degree $p$ and precision $\delta$ as
\begin{align*}
m = \left\lceil\log_2 \left(L\left(\frac{C_p\|V^{(p+1)}\|_\infty}{\delta}\right)^{\frac{1}{p+1}}\right) \right\rceil.
\end{align*}
For arbitrarily fixed $x\in [0,L]$, let $h(x)$ solve
$$
C_p \max_{y\in [x,x+h(x)]} |V^{(p+1)}(y)| (h(x))^{p+1} = \delta.
$$
Note that $h(x)$ is the maximal interval size at $x$ such that the error is upper bound by $\delta$. Then the number of intervals $\tilde M$ can be approximated by 
\begin{align*}
\tilde M \approx \int_0^L \frac{1}{h(x)} dx \ge \int_0^L \left(\frac{C_p|V^{(p+1)}(x)|}{\delta}\right)^{\frac{1}{p+1}} dx =: \hat M_{p}. 
\end{align*}
In general, we cannot take the number of polynomial $\tilde M$ as the above lower bound $\hat M_{p}$, and  $\tilde M$ (possibly depending on the set of knots $\{kL/M\}_{k=1}^{M-1}$) is numerically calculated by some algorithms, for example, the ones we propose in Appendix \ref{subsec:appA5}. However, we observe that $\hat M_{p} \le \tilde M \le c \hat M_{p}\le M$ for some constant $c>1$, which becomes close to $1$ as $\delta\to 0$ (see Appendix \ref{subsec:appB2}). Therefore, we can approximately use $\hat M_{p}$ for the asymptotic gate count. 

\subsection{Arithmetic operations and phase kickback}
\label{subsec:appA4}

According to \cite{Kassal.2008, Jones.2012, Babbush.2018, Haner.2018, Sanders.2020}, we provide the gate implementation and estimate the gate count and the required number of ancillary qubits in each step mentioned in Sect.~\ref{sec:methods}. 

For step 1, H\"aner et al. \cite{Haner.2018} suggested using a comparator with a given integer by a carry circuit \cite{Haner.2017}, an incrementer by \cite{Gidney2015}, and the uncomputation for each interval of the piecewise function. 
The comparator in \cite{Haner.2017} on an $n$-qubit system register uses $4(n-1)$ Toffoli gates, (at most) $2n+2$ CNOT gates, and (at most) $2(n-1)$ NOT gates with $n$ uninitialized ancillas and a zero-initialized ancilla. 
On the other hand, the incrementer by \cite{Gidney2015} uses $4(n-1)$ Toffoli gates, $10(n-1)$ CNOT gates, and $2n-1$ NOT gates with $n$ uninitialized ancillas, which implies that step 1 requires $12(n-1)(\tilde M-1)$ Toffoli gates, $(14n-6)(\tilde M-1)$ CNOT gates, and $(6n-5)(\tilde M-1)$ NOT gates with $n$ uninitialized ancillas, a zero-initialized ancilla and a zero-initialized register of $\lceil\log_2 \tilde M\rceil$ qubits. 

For step 2, according to the QROM circuit in \cite[Fig. 10]{Babbush.2018}, loading classical data to a data register of $b_{\text{data}}$ qubits requires $2(\tilde M-1)$ Toffoli gates, (at most) $b_{\text{data}}\tilde M + \tilde M -1$ CNOT gates, and $2(\tilde M-1)$ NOT gates with $\lceil\log_2 \tilde M\rceil +1$ initialized ancillas and a zero-initialized register of $b_{\text{data}}$ qubits. 
Owing to the specific structure, the circuit can be simplified to $4(\tilde M-1)$ T gates, $(b_{\text{data}}+\mathcal{O}(1))\tilde M$ CNOT gates, and $\mathcal{O}(\tilde M)$ Hadamard or S gates \cite[Fig. 4]{Babbush.2018}. 

In our formulation, we regard the system register as an integer register whose value is proportional to the grid points with a multiplier $\Delta x = L/N$. Then, the operations in step 1 can be executed with an integer comparator at a relatively cheap cost. 
By noting 
\begin{align*}
f_\ell(x_j) &= a_p^{(\ell)} x_j^p + \cdots + a_1^{(\ell)} x_j + a_0^{(\ell)}\\
&= \left(a_p^{(\ell)}(\Delta x)^p\right) j^p + \cdots + (a_1^{(\ell)}\Delta x) j + a_0^{(\ell)},
\end{align*}
for each $\ell=0,1,\ldots,\tilde M-1$, we need to load $p+1$ coefficients to the data registers, which store non-integer values. 
For simplicity, we follow the idea of \cite{Haner.2018} to consider fixed-point arithmetic and represent a number $y$ in a non-integer register using $b_{\text{tot}}$ binary bits as 
$$
y = y_{b_{\text{tot}}-1}\cdots y_{b_{\text{dec}}}.y_{b_{\text{dec}}-1}\cdots y_0,
$$
where $b_{\text{dec}}$ is the number of bits after the decimal point. Henceforth, we set $b_{\text{data}}=b_{\text{tot}}$, which will be estimated later. In other words, we use the same $(b_{\text{tot}},b_{\text{dec}})$ binary representation for all the data registers. 
To sum up, step 2 requires $2(p+1)(\tilde M-1)$ Toffoli gates, (at most) $(p+1)(b_{\text{tot}}\tilde M + \tilde M -1)$ CNOT gates, and $2(p+1)(\tilde M-1)$ NOT gates with $\lceil\log_2 \tilde M\rceil +1$ initialized ancillas and zero-initialized registers of totally $(p+1)b_{\text{tot}}$ qubits. 

For step 3, we first introduce another non-integer register for the grid points: $\ket{x(j)}$ with the same parameters $(b_{\text{tot}},b_{\text{dec}})$. By noting that $x(j)=x_j=j\Delta x$, the preparation of such a register boils down to a fixed-point multiplication of $j/N$ and $L$. According to \cite[Appendix B]{Haner.2018}, such a fixed-point multiplication uses controlled addition circuits on $k$ qubits for $k=b_{\text{dec}}+1,\ldots,b_{\text{tot}}$ and another controlled addition circuits on $k$ qubits for $k=b_{\text{tot}}-b_{\text{dec}},\ldots,b_{\text{tot}}-1$ with $b_{\text{tot}}$ zero-initialized ancillas. By employing the controlled addition circuits in \cite{Takahashi.2010}, the fixed-point multiplication uses $\frac{3}{2}b_{\text{tot}}^2 + 3b_{\text{dec}}b_{\text{tot}} - 3b_{\text{dec}}^2 +\frac{9}{2}b_{\text{tot}}-3b_{\text{dec}}$ Toffoli gates and $\mathcal{O}(b_{\text{tot}}^2)$ CNOT gates. 
Here, we mention that due to the truncation of the lower qubits if $n>b_{\text{dec}}$, the values in the register $\ket{x(j)}$ have an error of at most $L/2^{b_{\text{dec}}}$.

Also, we use the above $(b_{\text{tot}},b_{\text{dec}})$ binary representation for the result register $\ket{\tilde V(x_j)}$. 
Then, by applying the Horner scheme, we find that the implementation of the piecewise polynomial in a result register requires $p$ fixed-point addition and $p$ fixed-point multiplication operations:
\begin{align*}
\ket{0} &\mapsto \ket{a_p^{(\ell)} x(j)+a_{p-1}^{(\ell)}} 
\mapsto \ket{a_p^{(\ell)} x(j)^2+a_{p-1}^{(\ell)} x(j)+a_{p-2}^{(\ell)}} \mapsto \cdots \mapsto \ket{a_p^{(\ell)} x(j)^p + \cdots + a_1^{(\ell)} x(j) + a_0^{(\ell)}}.
\end{align*}
Then, step 3 requires $p\left(\frac{3}{2}b_{\text{tot}}^2 + 3b_{\text{dec}}b_{\text{tot}} - 3b_{\text{dec}}^2 +\frac{13}{2}b_{\text{tot}}-3b_{\text{dec}}-1\right)$ Toffoli gates and $\mathcal{O}(p b_{\text{tot}}^2)$ CNOT gates with $p+1$ zero-initialized result registers containing $(p+1)b_{\text{tot}}$ qubits. 
Here, we do not include the gate count of generating $\ket{x(j)}$ as there may be a better way to do it, and this yields at most a pre-factor to the above estimation. 
Moreover, we mention that the enduring ancillary qubits can be reduced to $2b_{\text{tot}}$ if one finds and employs an efficient in-place fixed-point multiplication. 
Furthermore, considering the truncation errors in the registers $\ket{x(j)}$ and $\ket{a_k^{(\ell)}}$, the accumulated error in the result register is upper bound by $c_p(1+\sum_{q=1}^p(|a_q|_\infty+1) L^q)/2^{b_{\text{dec}}}$ where $|a_p|_\infty$ denotes the maximum of $|a_p^{(\ell)}|$ over all the intervals $[\tilde x_\ell, \tilde x_{\ell+1})$, $\ell=0,1,\ldots,\tilde M-1$ and $c_p$ is some constant that possibly depends on $p$. 

Now, we estimate the parameters $b_{\text{dec}}$ and $b_{\text{tot}}$ for a given precision $\delta/2$. For the number of bits after the decimal point, the above estimate immediately gives
\begin{align*}
b_{\text{dec}} = 1+\left\lceil\log_2 c_p(1+\sum_{q=1}^p(|a_q|_\infty+1) L^q)/\delta \right\rceil,
\end{align*}
which is sufficient. On the other hand, for the integer part of the register, we have to avoid overflow in the fixed-point arithmetic operations to guarantee that the values in the result register approximate the calculated piecewise polynomial function. 
Then, the integer part should be able to describe values up to $\sum_{q=0}^p |a_q|_\infty L^q$, which implies that a total number of
\begin{align*}
b_{\text{tot}} = 1 + \left\lceil\log_2 c_p(1+\sum_{q=1}^p(|a_q|_\infty+1) L^q)/\delta \right\rceil +\left\lceil\log_2 \sum_{q=0}^p|a_q|_\infty L^q \right\rceil
\end{align*}
is sufficient to derive a precision $\delta/2$ in the result register. 
In particular, if $p=1$, we have its upper bound as 
$$
b_{\text{tot}} = 1 + \left\lceil\log_2 (L\|V^{(1)}\|_\infty + \|V\|_\infty) \right\rceil + \left\lceil\log_2 (1 + L + L\|V^{(1)}\|_\infty)\right\rceil + \lceil\log_2 1/\delta\rceil.
$$
We remark that the number of qubits in the register for the result depends on the length $L$, some $L^\infty$ norms of the potential, as well as its derivatives, and the desired precision $\delta$, but is independent of the grid parameter $n$. 
Moreover, if $x_j$ happens to have a $b_{\text{tot}}$-bit binary representation for all $j=0,1,\ldots,N-1$, then there is no error in the register $\ket{x(j)}$, and we find that 
$$
b_{\text{dec}} = 1+ \left\lceil\log_2 \sum_{q=0}^p L^q/\delta \right\rceil
$$
is enough. 

For step 4, we use the $(b_{\text{tot}}+b_{\text{ext}})$-bit binary representation for $\ket{\gamma(\xi)}$, where $b_{\text{ext}}$ is some parameter for extra qubits. 
According to \cite{Jones.2012}, the phase kickback is approximated by a modular adder between the result register (we can regard it as an integer register this time) and a prepared phase kickback register, which is derived by loading an integer smaller than $2^{b_{\text{tot}}+b_{\text{ext}}}$ and applying a QFT circuit. 
Therefore, step 4 requires a QFT on $(b_{\text{tot}}+b_{\text{ext}})$ qubits, $2(b_{\text{tot}}+b_{\text{ext}})-1$ Toffoli gates, and $\mathcal{O}(b_{\text{tot}}+b_{\text{ext}})$ CNOT gates with $(b_{\text{tot}}+2b_{\text{ext}})$ zero-initialized ancillas for the registers and addition by the adder in \cite{Takahashi.2010}. In order to guarantee a precision $\delta/4$, the parameter $b_{\text{ext}}$ should satisfy
\begin{align*}
b_{\text{ext}} = b_{\text{dec}}-b_{\text{tot}} + \left\lceil \log_2 \frac{8\pi\|V\|_\infty}{\delta}\right\rceil.
\end{align*}

Finally, step 5 includes the uncomputation, which doubles the gate count and circuit depth in steps 1--3. 

If we assume the potential $V$ is sufficiently smooth with bounded maximum norms and pay attention only to the parameters $p$, $n$, and $\delta$, then APK uses $\mathcal{O}\left((n+p)/\delta^{1/(p+1)}+p(p+\log (1/\delta))^2\right)$ Toffoli gates, $\mathcal{O}\left((n+p)/\delta^{1/(p+1)}+p(p+\log (1/\delta))^2\right)$ Clifford gates, and $\mathcal{O}\left((\log (1/\delta))^2\right)$ phase gates with 
$\mathcal{O}(1+p(p+\log (1/\delta)))$ initialized ancillary qubits. 

\subsection{Spline method}
\label{subsec:appA5}

We end this section with a proposal of classical algorithms to find the coefficients of a suitable piecewise polynomial approximation under a given precision. 
We note that the problem of finding piecewise polynomial approximations for a given function under a given precision does not admit uniqueness as one has many parameters such as the number of sub-intervals $\tilde M$ and the degrees of the polynomial $p_\ell$, $\ell=1,2,\ldots,\tilde M$ in each sub-interval. 
Moreover, even if the above parameters are given, finding an optimal piecewise polynomial approximation is not unique because we have selections on the choice of the norm (used to evaluate the difference between the approximate and the original function). 

Here, we use the maximum/supremum norm and consider the spline interpolation to derive the piecewise polynomial, which is efficient and avoids the problem of Runge's phenomenon. 
We mention other possible efficient methods, for example, the least square method or minimax approximation based on solving a minimization problem. Here, we choose the spline method since there are theoretical error estimations with known optimal pre-constants \cite{Agarwal.1991, Hall.1976, Dubeau.1996}, which helps us avoid numerically solving a minimization problem if the polynomial degree is specified in advance and enables us to estimate the parameters analytically. In particular, evaluations of the division parameter $m$ and number of intervals $\tilde M$ yield the gate count/depth of the quantum circuit. 

Our first proposal is for the case where the polynomial degree is a given constant. The procedure is listed as follows:

\noindent \underline{Algorithm 1}
\begin{enumerate}
\item Given a function $V$ with domain $[0,L]$, precision $\delta>0$, polynomial degree $p$, calculate the $(p+1)$-th derivative of $V$ and the optimal pre-constant $C_p$ which can be found in \cite{Agarwal.1991, Hall.1976, Dubeau.1996}. 
\item Calculate the division parameter $m\in \mathbb{N}$ explicitly by 
\begin{equation*}
m = \left\lceil\log_2 \left(L\left(\frac{C_p\|V^{(p+1)}\|_\infty}{\delta}\right)^{\frac{1}{p+1}}\right) \right\rceil, 
\end{equation*}
which defines $2^m$ sub-intervals $I_j=[jL/2^m, (j+1)L/2^m]$, $j=0,1,\ldots,2^m-1$.  
\item (i) Let $j_1=0$, $j_2=1$. 
(ii) Let $j_2 \leftarrow j_2+1$, calculate the maximum norm of the $(p+1)$-th derivative of $V$ in the interval $[j_1L/2^m, j_2L/2^m]$ denoted by $\|V_{j_1,j_2}^{(p+1)}\|_\infty$. If 
\begin{equation*}
C_p\|V_{j_1,j_2}^{(p+1)}\|_\infty \left((j_2-j_1)L/2^m\right)^{p+1}\le \delta, 
\end{equation*}
then go to the beginning of (ii). Otherwise, save $j_2-1$ and update $j_1\leftarrow j_2-1$, $j_2\leftarrow j_1+1$, and go to the beginning of (ii). 
(iii) Continue (ii) and stop when $j_2=2^m+1$. 
Then, a rough division of $\tilde M\le 2^m$ sub-intervals is obtained according to the saved numbers. 
\item Using the $\tilde M+1$ endpoints of the derived $\tilde M$ sub-intervals as knots, calculate the coefficients of the polynomial in each sub-interval by applying the $p$-th degree spline method. 
\end{enumerate}
As we know, by quantum computing, we divide the interval into $2^m$ sub-intervals, and Steps 1--2 calculate explicitly the minimal $m$ such that the error is still bound by $\delta$ in the ``worst" (most rapidly varying) sub-interval of $V$. 
Since the error by spline interpolation is very small in some mildly varying parts of $V$, Step 3 combines several sub-intervals so that the error in every new sub-intervals would be close to the given error bound $\delta$. 
Of course, one can skip Step 3 and set $\tilde M=2^m$ in Step 4 by using the sub-intervals $I_j$ in Step 2, which yields a uniform division, but Step 3 gives a more flexible non-uniform division, which reduces the number of sub-intervals. 

Next, we propose a more involved algorithm where the polynomial degree can differ in each sub-interval. As mentioned above, there is never uniqueness in such a varying degree case, and we need to introduce some object function to obtain one unique solution. 

\noindent \underline{Algorithm 2}
\begin{enumerate}
\item Given a function $V$ with domain $[0,L]$, precision $\delta>0$, grid parameter $n$ and admissible sets $\mathcal{M} = \{m_{\text{min}}, \ldots, m_{\text{max}}\}$, $\mathcal{P} = \{p_{\text{min}}, \ldots, p_{\text{max}}\}$ with $m_{\text{max}}\le n$. 

\item For each $m\in \mathcal{M}$, define $2^m$ sub-intervals $I_k=[kL/2^m, (k+1)L/2^m]$, $k=0,1,\ldots,2^m-1$. Then, calculate the polynomial degree for each $k=0,1,\ldots,2^m-1$ by solving the following constrained (minimization) problem:
\begin{equation}
\label{app:minp1}
p_k = p_k(m) := \min\left\{p\in \mathcal{P}: \, C_p\|V_{k,k+1}^{(p+1)}\|_\infty\left(L/2^m\right)^{p+1}\le \delta \right\}, 
\end{equation}
where $C_p$ denotes the optimal pre-constant in \cite{Agarwal.1991, Hall.1976, Dubeau.1996}, and $\|V_{k,k+1}^{(p+1)}\|_\infty$ is the same notation as that in Algorithm 1. 

\item For each $m\in \mathcal{M}$, execute the following steps.  

(i) Let $j_1=0$, $j_2=1$. 
(ii) Let $j_2 \leftarrow j_2+1$. If $p_{j_2-1}(m)\not=p_{j_1}(m)$, then save $j_2-1$ and update $j_1\leftarrow j_2-1$, $j_2\leftarrow j_1+1$, and go to the beginning of (ii). Else if 
\begin{equation*}
C_{p_{j_1}}\|V_{j_1,j_2}^{(p_{j_1}+1)}\|_\infty \left((j_2-j_1)L/2^m\right)^{p_{j_1}+1}\le \delta, 
\end{equation*}
then go to the beginning of (ii). Otherwise, save $j_2-1$ and update $j_1\leftarrow j_2-1$, $j_2\leftarrow j_1+1$, and go to the beginning of (ii). 
(iii) Continue (ii) and stop when $j_2=2^m+1$. 
Then, a rough division of $\tilde M = \tilde M(m)\le 2^m$ sub-intervals is obtained according to the saved numbers. 

\item Based on the parameter $m\in \mathcal{M}$, the calculated $p_{k}(m)$ and $\tilde M(m)$, define an objective function $F(m;n)=\tilde F(m,p_{k}(m),\tilde M(m);n)$ where $\tilde F$ could be CNOT count for the quantum circuit implementing the piecewise polynomial function with parameters $n,m,\tilde M$ and $p_k$. Solve the minimization problem:
\begin{equation}
\label{app:minp2}
m^\ast = \mathrm{argmin}_{m\in \mathcal{M}} F(m;n).
\end{equation}

\item According to the parameters $m^\ast$, $\tilde M^\ast = \tilde M(m^\ast)$, and $p_k^\ast = p_k(m^\ast)$, $k=0,1,\ldots,2^{m^\ast}-1$, calculate the coefficients of the $\tilde M^\ast$ polynomials by applying the $p_{k_j}$-th degree spline method for each interval with index $j=0,1,\ldots,\tilde M^\ast-1$. 
\end{enumerate}
Algorithm 2 is a specific classical algorithm to find a suitable piecewise polynomial approximation to a given function corresponding to its quantum implementation. 
Here, we need to introduce two finite admissible sets $\mathcal{M}$, $\mathcal{P}$ to guarantee the well-posedness of the minimization problems \eqref{app:minp1} and \eqref{app:minp2} (unique existence of the solution to each problem). 
Moreover, for the given function $V$ and precision $\delta$, $m_{\text{max}}$ and $p_{\text{max}}$ must be chosen sufficiently large, or there could be no solution to the minimization problems. 
The comparison between Algorithms 1 and 2 will be provided in Appendix \ref{subsec:appB1}. 

\section{Comparison of different methods}
\label{sec:appB}

\subsection{Comparison of classical algorithms}
\label{subsec:appB1}

In this subsection, we compare the classical algorithms based on the spline methods by using an example. 
Assume that we are given the modified Coulomb potential $V(x)=A/\sqrt{a^2+(x-L/2)^2}$, $x\in [0,L]$ with $A=1$, $L=20$, and $a^2=0.5$. The numerically calculated coarse-graining parameter $m$ and sub-interval number $\tilde M$ by several different spline interpolations are listed in Table \ref{tab:appB-1}. 
\begin{table}[htb]
\centering
\caption{Parameters $m$, $\tilde M$ by different splines with respect to $\delta$ 
for modified Coulomb potential with $A=1$, $a^2=0.5$, $L=20$.}
\label{tab:appB-1}
\scalebox{1.0}[1.0]{
\begin{tabular}{l|ccccc}
\hline
$m$, $\tilde M$ & $\delta=10^{-1}$ & $\delta=10^{-2}$ & $\delta=10^{-3}$ & $\delta=10^{-4}$ & $\delta=10^{-6}$ \\
\hline
Uniform linear$^{[a]}$ & 6, 64 & 7, 128 & 9, 512 & 11, 2048 & 14, 16384 \\
Linear$^{[b]}$ & 6, 12 & 7, 26 & 9, 70 & 11, 216 & 14, 2072 \\
Quadratic$^{[c]}$ & 5, 8 & 6, 16 & 7, 30 & 8, 60 & 11, 270 \\
Cubic$^{[d]}$ & 6, 12 & 6, 18 & 7, 28 & 8, 48 & 10, 128 \\
Cubic Hermite$^{[e]}$ & 5, 8 & 6, 12 & 7, 18 & 7, 30 & 9, 94 \\
Optimized degree-varying$^{[f]}$ & 6, 12 & 7, 26 & 8, 70 & 8, 124 & 10, 594 \\
\hline
\end{tabular}
}
\begin{flushleft}
\footnotesize $[a]$ Algorithm 1 without Step 3 for $p=1$ and $C_1=1/8$ (linear interpolation \cite{Agarwal.1991, Hall.1976}).  

$[b]$ Algorithm 1 for $p=1$ and $C_1=1/8$ (linear interpolation \cite{Agarwal.1991, Hall.1976}). 

$[c]$ Algorithm 1 for $p=2$ and $C_2=2/81$ (two-point Hermite interpolation \cite{Agarwal.1991}). 

$[d]$ Algorithm 1 for $p=3$ and $C_3=5/384$ (cubic spline \cite{Hall.1976}).

$[e]$ Algorithm 1 for $p=3$ and $C_3=1/384$ (cubic Hermite spline \cite{Agarwal.1991, Hall.1976}).

$[f]$ Algorithm 2 for $\mathcal{P}=\{1,2,3\}$, $\mathcal{M}=\{2,\ldots,n\}$ and $C_1=1/8$, $C_2=2/81$, $C_3=1/384$ with $n=19$. Here, the objective function is set to be CNOT count corresponding to PPP (see Appendix \ref{subsec:appA3}). 
\end{flushleft}
\end{table}
Comparing $\tilde M$ for [a] and [b], we find Step 3 in Algorithm 1 (non-uniform division) heavily reduces the number of intervals. In addition, the results for [d] and [e] show that the cubic Hermite spline outperforms the cubic spline of type I/type II \cite{Hall.1976} since its optimal pre-constant is smaller. Moreover, focusing on [b], [c], and [e], we find that a higher degree spline interpolation gives smaller parameters $m$ and $\tilde M$. Furthermore, Algorithm 2 provides ``worse" parameters than Algorithm 1 by noting [e] and [f]. However, it gives a better CNOT count because the quantum circuit for a higher degree polynomial is more expensive, and it balances the polynomial degree regarding the cost of quantum implementation.
Using the quantum circuit of PPP, we can calculate the CNOT count of the quantum circuit analytically, and the result is given in Table \ref{tab:appB-2}. 
\begin{table}[htb]
\centering
\caption{CNOT count by different splines and PPP with respect to $\delta$ 
for modified Coulomb potential with $A=1$, $a^2=0.5$, $L=20$ and $n=19$.}
\label{tab:appB-2}
\scalebox{1.0}[1.0]{
\begin{tabular}{l|ccccc}
\hline
CNOT count & $\delta=10^{-1}$ & $\delta=10^{-2}$ & $\delta=10^{-3}$ & $\delta=10^{-4}$ & $\delta=10^{-6}$ \\
\hline
Uniform linear$^{[a]}$ & 20538 & 54610 & 350546 & 2059282 & 26311098 \\
Linear$^{[b]}$ & 3586 & 10750 & 47334 & 216290 & 3326026 \\
Quadratic$^{[c]}$ & 11584 & 25752 & 52484 & 113504 & 638948 \\
Cubic$^{[d]}$ & 239908 & 366352 & 579900 & 1009100 & 2749516 \\
Cubic Hermite$^{[e]}$ & 154996 & 239908 & 368120 & 622256 & 2001456 \\
Optimized degree-varying$^{[f]}$ & 3586 & 10750 & 44790 & 121002 & 1015862 \\
\hline
\end{tabular}
}
\begin{flushleft}
\footnotesize $[a]$ Algorithm 1 without Step 3 for $p=1$ and $C_1=1/8$ (linear interpolation \cite{Agarwal.1991, Hall.1976}).  

$[b]$ Algorithm 1 for $p=1$ and $C_1=1/8$ (linear interpolation \cite{Agarwal.1991, Hall.1976}). 

$[c]$ Algorithm 1 for $p=2$ and $C_2=2/81$ (two-point Hermite interpolation \cite{Agarwal.1991}). 

$[d]$ Algorithm 1 for $p=3$ and $C_3=5/384$ (cubic spline \cite{Hall.1976}).

$[e]$ Algorithm 1 for $p=3$ and $C_3=1/384$ (cubic Hermite spline \cite{Agarwal.1991, Hall.1976}).

$[f]$ Algorithm 2 for $\mathcal{P}=\{1,2,3\}$, $\mathcal{M}=\{2,\ldots,n\}$ and $C_1=1/8$, $C_2=2/81$, $C_3=1/384$ with $n=19$. Here, the objective function is set to be CNOT count corresponding to PPP (see Appendix \ref{subsec:appA3}). 
\end{flushleft}
\end{table}
With a fixed grid parameter $n=19$, Table \ref{tab:appB-2} confirms the better performance of the non-uniform spline beyond the uniform one and that the cubic Hermite spline yields fewer CNOT gates than the cubic spline of type I/type II. However, due to the polynomial factor $n^p$ (exponentially increasing for $p$) in gate count, it is not correct that the higher degree we choose, the better performance we gain. The best polynomial degree in Algorithm 1 depends on $n$, $\delta$, $V$, and the spline method one uses. In the above specific case, the quadratic spline is the best for sufficiently small $\delta$. 
A slightly involved issue is the comparison between Algorithms 1 and 2. Algorithm 2 gives a reasonable spline without knowledge of degree $p$, but it does not outperform Algorithm 1 sometimes. 
We choose the smallest possible degree in each sub-interval in Algorithm 2 to reduce the difficulty of solving the minimization problem, which prevents us from the trade-off by increasing the degree and decreasing the total number of sub-intervals. 
Of course, Algorithm 2 with $\mathcal{P} = \{p_0\}$ outperforms Algorithm 1 for a given degree $p_0$ because it addresses the trade-off by increasing the division parameter $m$ and decreasing the total number of sub-intervals. 

On the other hand, we give a remark on the optimal pre-constants in spline interpolations. According to the spline method one uses, such constant varies and can be found in \cite{Agarwal.1991, Hall.1976, Dubeau.1996}. The word ``optimal" means that the error estimate (i.e., the second inequality in \eqref{app:err-est}) would become equality in the worst case. The optimality is closely related to the asymptotic behavior as $\delta\to 0$, which we illustrate by listing the numerical maximum errors in Table \ref{tab:appB-3}. 
We find that the maximum error between the spline interpolation and the original function is larger than $\delta/2$ as $\delta\le 0.01$ and tends to $\delta$ as $\delta\to 0$, which indicates that our algorithm based on these theoretical optimal constants would be efficient in high-precision cases. In contrast, it may give higher precision than desired for a large given $\delta$ at the cost of using more resources.  
\begin{table}[htb]
\centering
\caption{Numerical maximum error by different splines with respect to $\delta$ 
for modified Coulomb potential with $A=1$, $a^2=0.5$, $L=20$.}
\label{tab:appB-3}
\scalebox{0.77}[0.77]{
\begin{tabular}{l|ccccccc}
\hline
Maximum error & $\delta=10^{-1}$ & $\delta=10^{-2}$ & $\delta=10^{-3}$ & $\delta=10^{-4}$ & $\delta=10^{-5}$ & $\delta=10^{-6}$ & $\delta=10^{-7}$ \\
\hline
Linear$^{[a]}$ & $0.403\times 10^{-1}$ & $0.811\times 10^{-2}$ & $0.973\times 10^{-3}$ & $0.980\times 10^{-4}$ & $0.991\times 10^{-5}$ & $0.997\times 10^{-6}$ & $0.999\times 10^{-7}$ \\
Quadratic$^{[b]}$ & $0.382\times 10^{-1}$ & $0.797\times 10^{-2}$ & $0.718\times 10^{-3}$ & $0.906\times 10^{-4}$ & $0.926\times 10^{-5}$ & $0.962\times 10^{-6}$ & $0.985\times 10^{-7}$ \\
Cubic Hermite$^{[c]}$ & $0.167\times 10^{-1}$ & $0.535\times 10^{-2}$ & $0.772\times 10^{-3}$ & $0.664\times 10^{-4}$ & $0.691\times 10^{-5}$ & $0.827\times 10^{-6}$ & $0.952\times 10^{-7}$ \\
\hline
\end{tabular}
}
\begin{flushleft}
\footnotesize 
$[a]$ Algorithm 1 for $p=1$ and $C_1=1/8$ (linear interpolation \cite{Agarwal.1991, Hall.1976}). 

$[b]$ Algorithm 1 for $p=2$ and $C_2=2/81$ (two-point Hermite interpolation \cite{Agarwal.1991}). 

$[c]$ Algorithm 1 for $p=3$ and $C_3=1/384$ (cubic Hermite spline \cite{Agarwal.1991, Hall.1976}).
\end{flushleft}
\end{table}

\subsection{Relation between numerically calculated number of sub-intervals and its theoretical lower bound}
\label{subsec:appB2}

We check the relation between the number of sub-intervals $\tilde M_p$ and its theoretical lower bound $\hat M_{p}$ for the polynomial degree $p=1,2,3$, the precision $\delta>0$ and two different functions. 

The first function is the modified Coulomb potential in Example \ref{exa:1} of Sect.~\ref{subsec:3-2}. 
Let $A=1$ and $L=20$ be constants. We consider several values of parameter $a^2$. 
For a given polynomial degree $p$, recall that $\tilde M_p$ is the number of sub-intervals obtained in Step 3 of Algorithm 1, while $\hat M_{p}$ is its theoretical lower bound defined in Appendix \ref{subsec:appA3}. We plot the ratio $\tilde M_p/\hat M_p$ as a function of $\delta$ under the different choices of $a^2=0.5, 0.1, 0.001$ in Fig.~\ref{appB:Fig1}. 
\begin{figure}
\centering
\resizebox{15cm}{!}{
\includegraphics[keepaspectratio]{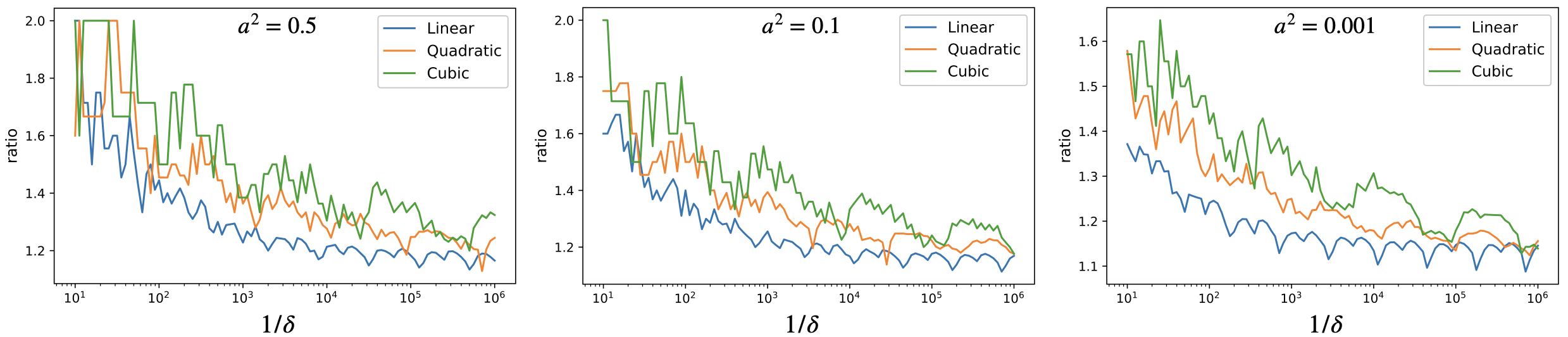}
}
\caption{Ratio between the number of sub-intervals $\tilde M_p$ by Algorithm 1 and its lower bound $\hat M_p$ against the reciprocal of precision $1/\delta$ for the modified Coulomb potential with fixed $A=1$, $L=20$ and varying parameter $a^2=0.5, 0.1, 0.001$. In each subplot, we demonstrate the change of the ratios for different spline interpolations as $\delta$ decreases. } 
\label{appB:Fig1}
\end{figure}
In the calculations of $\tilde M_p$ and $\hat M_p$, we use $C_1=1/8$, $C_2=2/81$, and $C_3=1/384$, which corresponds to the two-point Hermite interpolations \cite{Agarwal.1991}. 
Besides, $\hat M_{p}$ is calculated by numerical integration (a Riemann sum) for given $p$ and $\delta$. 

The second function is an oscillating one in Example \ref{exa:2} of Sect.~\ref{subsec:3-2}. 
Let $A=1$, $\omega=2$, and $L=20$ be constants. Again, we plot the ratio $\tilde M_p/\hat M_p$ as a function of $\delta$ under the different choices of $a=1, 0.1, 0.001$ in Fig.~\ref{appB:Fig2}. 
\begin{figure}
\centering
\resizebox{15cm}{!}{
\includegraphics[keepaspectratio]{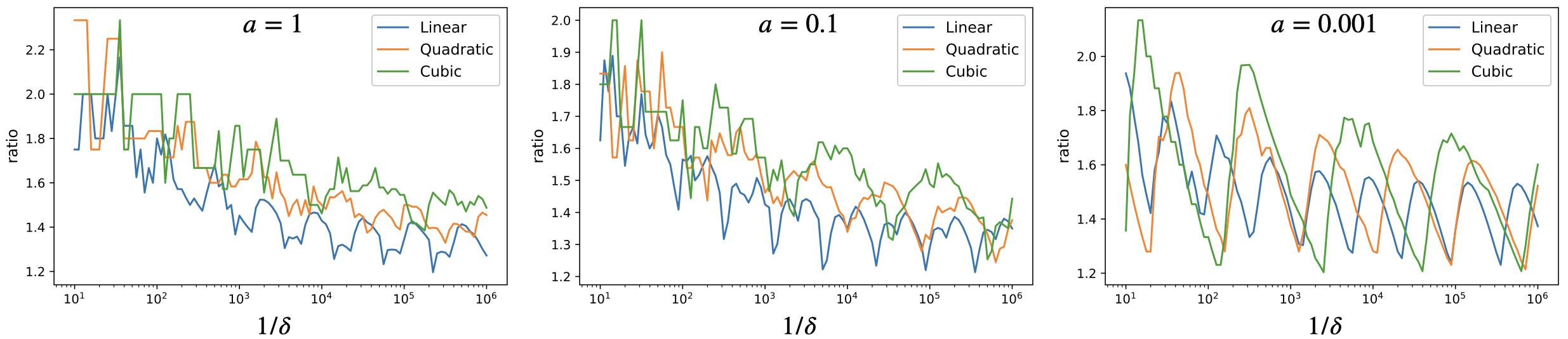}
}
\caption{Ratio between the number of sub-intervals $\tilde M_p$ by Algorithm 1 and its lower bound $\hat M_p$ against the reciprocal of precision $1/\delta$ for an oscillating function with fixed $A=1$, $\omega=2$, $L=20$ and varying parameter $a=1, 0.1, 0.001$. In each subplot, we demonstrate the change of the ratios for different spline interpolations as $\delta$ decreases. } 
\label{appB:Fig2}
\end{figure}
Figs.~\ref{appB:Fig1}--\ref{appB:Fig2} show that for sufficiently small $\delta$, the lower bound $\hat M_{p}$ gives a good estimation of the number of intervals $\tilde M_p$ up to a constant $c$ (i.e., $\tilde M_p\le c\hat M_{p}$). Moreover, although such a constant varies for different functions and precision bounds, it is smaller than $7/3$ and becomes closer to $1$ for smaller $\delta$.

\subsection{Analytical comparison of WAL, LIU, and PPP}
\label{subsec:appB4}

In the case study, we find that PPP with $p=2$ uses fewer CNOT gates than PPP with $p=3$ in practical settings. Here, we clarify when such observation is valid under some assumptions. 
We assume $V\in C^{q}[0, L]$ for some $q\in \mathbb{N}$ and define $\tilde V_p(x):=(C_p|V^{(p+1)}(x)|)^{1/(p+1)}$ for $x\in [0, L]$ and $p=0,1,\ldots,q-1$, where $C_p$ is the optimal constant for spline interpolations \cite{Agarwal.1991}. For each $p$, we introduce several constants:
$$
C_{p,1} := \|\tilde V_{p}\|_{L^1(0,L)} = \int_0^L |\tilde V_{p}(x)| dx, \quad C_{p,\infty} := L \|\tilde V_{p}\|_{L^\infty(0,L)} = L \|\tilde V_p\|_{\infty}.
$$
Using these constants, we can express $m_p$ and $\hat M_p$ in Appendix \ref{subsec:appA3} by 
\begin{equation}
\label{app:eq-B1}
m_p = \left\lceil\log_2 \left(C_{p,\infty}\delta^{-1/(p+1)}\right)\right\rceil, \quad \hat M_p = \lceil C_{p,1}\delta^{-1/(p+1)}\rceil.
\end{equation}
We note that $C_{p,1}=C_{p,\infty}=0$ if and only if $V$ is a $p$-th degree polynomial in $[0, L]$. In such trivial cases, we can take $\tilde M_p=1$, and the implementation by PPP is precise and exhibits a lower gate count compared to WAL. Therefore, we mainly consider the cases that $C_{p,1}$ and $C_{p,\infty}$ do not vanish. From Eq.~\ref{app:eq-B1} and the facts that $C_{p,1}, C_{p,\infty}$ are independent of $\delta$, we find that $\tilde M_p\ge \hat M_p=\mathcal{O}(\delta^{-1/(p+1)})$ for sufficiently small $\delta$ provided that $C_{p,1}$ is bounded from below, which implies that for a given function $V$ and a sufficiently small $\delta$, a non-uniform division has a similar performance as a uniform division regarding $\delta$.  
According to the discussion in Appendix \ref{subsec:appB2}, we further assume that there exists a $\delta$-independent constant $c\ge 1$ such that $\hat M_p\le \tilde M_p\le c\hat M_p$ holds for any $p\in \mathbb{N}$. Thus, we have
$$
\tilde C := \frac{\tilde M_2}{\tilde M_3}\delta^{1/12}\in \left[c^{-1}\frac{C_{2,1}}{C_{3,1}}, c\frac{C_{2,1}}{C_{3,1}}\right].
$$
This indicates that $\tilde C$ is upper and lower bound by some constants that are independent of $\delta$. Using the analytic CNOT count for each method, for a fixed precision $\delta>0$, we obtain 

\noindent -- $CC_{\text{WAL}}(n) = 2^{n}-2$, $n\in [1, m_0]$,

\noindent -- $CC_{\text{LIU}}(n) = 2(n-m_1)(2^{m_1}+2m_1^2-2)+2^{m_1}-2$, $n\in [m_1, m_0]$,

\noindent -- $CC_{\text{mLIU}}(n) = (n-m_1)(3\cdot 2^{m_1}-4)+2^{m_1}-2$, $n\in [m_1, m_0]$,

\noindent -- $CC_{\text{PPP1}}(n) = (2n+8m_1^2)(\tilde M_1-1)$, $n\in [m_1, m_0]$,

\noindent -- $CC_{\text{PPP2}}(n) = n(n-1)+(2n+4n(n-1)+8m_2^2)(\tilde M_2-1)$, $n\in [m_2, m_0]$,

\noindent -- $CC_{\text{PPP3}}(n) = n(n-1)+4n(n-1)(n-2)/3+(2n+4n(n-1)+10n(n-1)(n-2)/3+8m_3^2)(\tilde M_3-1)$, $n\in [m_3, m_0]$.

\noindent First, we compare PPP with different polynomial degrees. Define the difference of CNOT count between PPP with $p=2$ and PPP with $p=3$ by
\begin{align*}
g(n) &:= CC_{\text{PPP2}}(n)-CC_{\text{PPP3}}(n)\\
&= (2n+4n(n-1)+8m_2^2)(\tilde M_2-\tilde M_3) -(8m_3^2-8m_2^2+\frac{10}{3}n(n-1)(n-2))(\tilde M_3-1) - \frac{4}{3}n(n-1)(n-2)\\
&= -\frac{1}{3}(10\tilde M_3-6)n^3 + (6\tilde M_3+4\tilde M_2-6)n^2 -\frac{1}{3}(14\tilde M_3+6\tilde M_2-12)n + 8m_2^2(\tilde M_2-1)-8m_3^2(\tilde M_3-1).
\end{align*}
Recall that $L/2^{m_p}$ and $\tilde M_p$ are the size of the sub-interval and the number of sub-intervals to guarantee the precision $\delta$ regarding $p$-th degree spline interpolation. It is natural to assume 
\begin{align}
\label{app:assump1}
m_1\ge m_2\ge m_3 \ge 1 \quad \mbox{ and }\quad \tilde M_2\ge \tilde M_3\ge 1,
\end{align} 
which implies $g(0)>0$. 
We calculate the derivative of $g$ with respect to $n$: 
\begin{align*}
g^\prime(n) &= -(10\tilde M_3-6)n^2 + (12\tilde M_3+8\tilde M_2-12)n -(14\tilde M_3+6\tilde M_2-12)/3.
\end{align*}
By assumption \eqref{app:assump1}, we can check
$$
(12\tilde M_3+8\tilde M_2-12)^2-4(10\tilde M_3-6)(14\tilde M_3+6\tilde M_2-12)/3>0.
$$
Then, letting $g^\prime(n)=0$, we have $g(n)$ takes its local minimum/maximum at 
\begin{align*}
n_{\pm}^\ast &= \frac{3\tilde M_3+2\tilde M_2-3 \pm \sqrt{(3\tilde M_3+2\tilde M_2-3)^2-(5\tilde M_3-3)(7\tilde M_3+3\tilde M_2-6)/3}}{5\tilde M_3-3}\\
&= \frac{(2\tilde C\delta^{-1/12}+3)\tilde M_3-3 \pm \sqrt{((2\tilde C\delta^{-1/12}+3)\tilde M_3-3)^2-(5\tilde M_3-3)((3\tilde C\delta^{-1/12}+7)\tilde M_3-6)/3}}{5\tilde M_3-3}.
\end{align*}
Therefore, $g(n)<0$ for $n\in [m_1, m_0]$ if and only if $n_+^\ast\le m_1$ and $g(m_1)<0$ or $n_+^\ast> m_1$ and $g(\min\{n_+^\ast, m_0\})<0$. 
By noting that $n_+^\ast$, $m_1$ depends on $\delta$, if we have
$$
\delta\in \Delta_V := \{\delta>0; n_+^\ast(\delta)\le m_1(\delta), g(m_1(\delta))<0\} \cup \{\delta>0; n_+^\ast(\delta)> m_1(\delta), g(\min\{n_+^\ast(\delta), m_0(\delta)\})<0\}, 
$$
then we obtain $g(n)<0$ for any $n\in [m_1(\delta), m_0(\delta)]$. 

Second, we compare WAL and PPP with $p=2$ in $n\in (0, m_1]$. We define the difference of CNOT count between these two methods by 
\begin{align*}
\bar g(n) &:= CC_{\text{WAL}}(n) - CC_{\text{PPP2}}(n)\\
&= 2^n-2-n(n-1)-(4n^2-2n+8m_2^2)(\tilde M_2-1). 
\end{align*}
We can verify $\bar g(0)<0$. We calculate the derivative of $\bar g$ with respect to $n$: 
\begin{align*}
\bar g^\prime(n) &= (\ln 2) 2^n -2n+1-(8n-2)(\tilde M_2-1)\\
&= (\ln 2) 2^n - (8\tilde M_2-6)n + (2\tilde M_2-1).
\end{align*}
Let $n_1^\ast$ and $n_2^\ast$ be the solutions to $\bar g^\prime(n) = 0$. Then, we can verify that $\bar g(n_1^\ast)<0$ and $\bar g(n_2^\ast)<0$ provided that $m_2\ge 2$ (i.e., $\delta<C_{2,\infty}^3/8$). 
Therefore, $\bar g(n)<0$ for $n\in (0,m_1]$ if $m_2\ge 2$ and $\bar g(m_1)<0$. If we define 
$$
\bar \Delta_V := \{0<\delta<C_{2,\infty}^3/8; \bar g(m_1(\delta))<0\}, 
$$
then we obtain $\bar g(n)<0$ for any $n\in (0, m_1(\delta)]$ provided that $\delta\in \bar \Delta_V$. 

Third, we also check the difference between LIU and WAL by defining
\begin{align*}
\tilde g(n) &:= CC_{\text{LIU}} - CC_{\text{WAL}}\\
&= 2(n-m_1)(2^{m_1}+2m_1^2-2)+2^{m_1}-2^n, \quad n\in [m_1, \infty).
\end{align*}
It is clear that $\tilde g(m_1)=0$, and we calculate the derivative
$$
\tilde g^\prime(n) = 2(2^{m_1}+2m_1^2-2)-(\ln 2) 2^n.
$$
Since $\tilde g^\prime(m_1)=(2-\ln 2)2^{m_1}+4m_1^2-4>0$ and $\tilde g^\prime(n)=0$ has only one solution, we conclude that $\tilde g(n)=0$ has two solutions $n_3^\ast=m_1$ and $n_4^\ast>m_1$. Therefore, if $\delta\in \tilde \Delta_V :=\{\delta>0; m_0(\delta)> n_4^\ast(\delta)\}$, then $\tilde g(n)\le 0$ for $n\in [n_4^\ast,m_0]$. In particular, we have $\tilde g(m_0)\le 0$. 

The above discussion explains the reason why we observe (i) and (iii) in Sect.~\ref{subsec:3-2}. By the definition of $\Delta_V$, $\bar\Delta_V$, and $\tilde \Delta_V$, we find that they depend only on the function $V$. For the given functions in Sect.~\ref{subsec:3-2}, we list these sets in Table \ref{tab:appB4-1} by numerical calculations. 
\begin{table}[htb]
\centering
\caption{Numerical evaluations of $\Delta_V$, $\bar\Delta_V$, and $\tilde \Delta_V$ for the functions in Sect.~\ref{subsec:3-2}. }
\label{tab:appB4-1}
\scalebox{0.59}[0.59]{
\begin{tabular}{l|lll}
\hline
 & $\Delta_V\cap \mathcal{A}_\Delta$ & $\bar\Delta_V\cap \mathcal{A}_\Delta$ & $\tilde\Delta_V\cap \mathcal{A}_\Delta$ \\
\hline
Ex. 1: $a^2=0.5$ & $(1\times 10^{-9},0.1)$ & $(1\times 10^{-9},0.1)$ & $(1\times 10^{-9},7.5\times 10^{-3})\cup (8.6\times 10^{-3}, 1.5\times 10^{-2})$ \\
Ex. 1: $a^2=0.1$ & $(1\times 10^{-9},0.1)$ & $(1\times 10^{-9},0.1)$ & $(1\times 10^{-9},1.9\times 10^{-2})\cup (2.4\times 10^{-2}, 3.8\times 10^{-2})$ \\
Ex. 1: $a^2=0.001$ & $(1\times 10^{-9},0.1)$ & $(1.4\times 10^{-9},4.0\times 10^{-9})\cup (5.6\times 10^{-9},0.1)$ & $(1\times 10^{-9},0.1)$ \\
Ex. 2: $a=1$ & $(1\times 10^{-9},0.1)$ & $(1\times 10^{-9},0.1)$ & $(1\times 10^{-9},3.4\times 10^{-2})$ \\
Ex. 2: $a=0.1$ & $(1\times 10^{-9},0.1)$ & $(1\times 10^{-9},0.1)$ & $(1\times 10^{-9},3.7\times 10^{-2})$ \\
Ex. 2: $a=0.001$ & $(1\times 10^{-9},5.8\times 10^{-9})\cup (7.9\times 10^{-9},1.2\times 10^{-8})$ & $(1\times 10^{-9},0.1)$ & $(1\times 10^{-9},3.9\times 10^{-2})$ \\
 & $\cup (1.3\times 10^{-8},9.1\times 10^{-8})\cup (1.1\times 10^{-7},0.1)$ &  &  \\
\hline
\end{tabular}
}
\end{table}
Here, $\mathcal{A}_\Delta:=(1\times 10^{-9}, 0.1)$ denotes a set of $\delta$. The minimal value is set as $1\times 10^{-9}$ because this is sufficiently small for practical applications, and it is very expensive for a classical computer to calculate a finer division than $2^{30}\approx 10^{9}$ sub-intervals in Algorithm 1 or 2. In contrast, although the estimation of $\tilde M_p$ by Algorithm 1 or 2 is time-consuming for small $\delta<1\times 10^{-9}$ on a classical computer, one can use the easily calculated theoretical bound $\hat M_p$ to obtain some subsets of $\Delta_V$, $\bar\Delta_V$, and $\tilde \Delta_V$ instead. 
From the first two columns of Table \ref{tab:appB4-1}, we verify that PPP with $p=2$ outperforms PPP with $p=3$ for $m_1 \le n\le m_0$, and WAL outperforms PPP for $0<n\le m_1$ provided that $1.1\times 10^{-7}< \delta < 0.1$.

\subsection{Comparison of methods on CNOT count/depth by Qiskit emulation}
\label{subsec:appB3}

In Sect.~\ref{sec:comp}, we gave a detailed comparison of WAL, LIU, and PPP on CNOT count by using analytical estimation for each method. 
In practice, such analytical evaluations deliver the corresponding theoretical upper bounds in the worst case. For example, we did not consider the possibility of canceling CNOT gates. 
In this subsection, we calculate the CNOT count and circuit depth of each method for several examples by using a quantum emulator (Qiskit \cite{Qiskit2023}). 
We check whether a hidden function-independent improvement of each method exists (specifically in CNOT count) and compare the circuit depth based on the native gates of the IBM quantum computer. 

First, we study the modified Coulomb potential as Example \ref{exa:1} with parameters $A=1$, $a^2=0.1$, and $L=20$ under a relatively high precision $\delta=0.001$ and a lower one $\delta=0.1$. 
For $n=m_1,m_1+1,\ldots,m_0$, the CNOT counts and depths for WAL, LIU, and PPP are demonstrated in Fig.~\ref{appB:Fig3}. 
\begin{figure}
\centering
\resizebox{13cm}{!}{
\includegraphics[keepaspectratio]{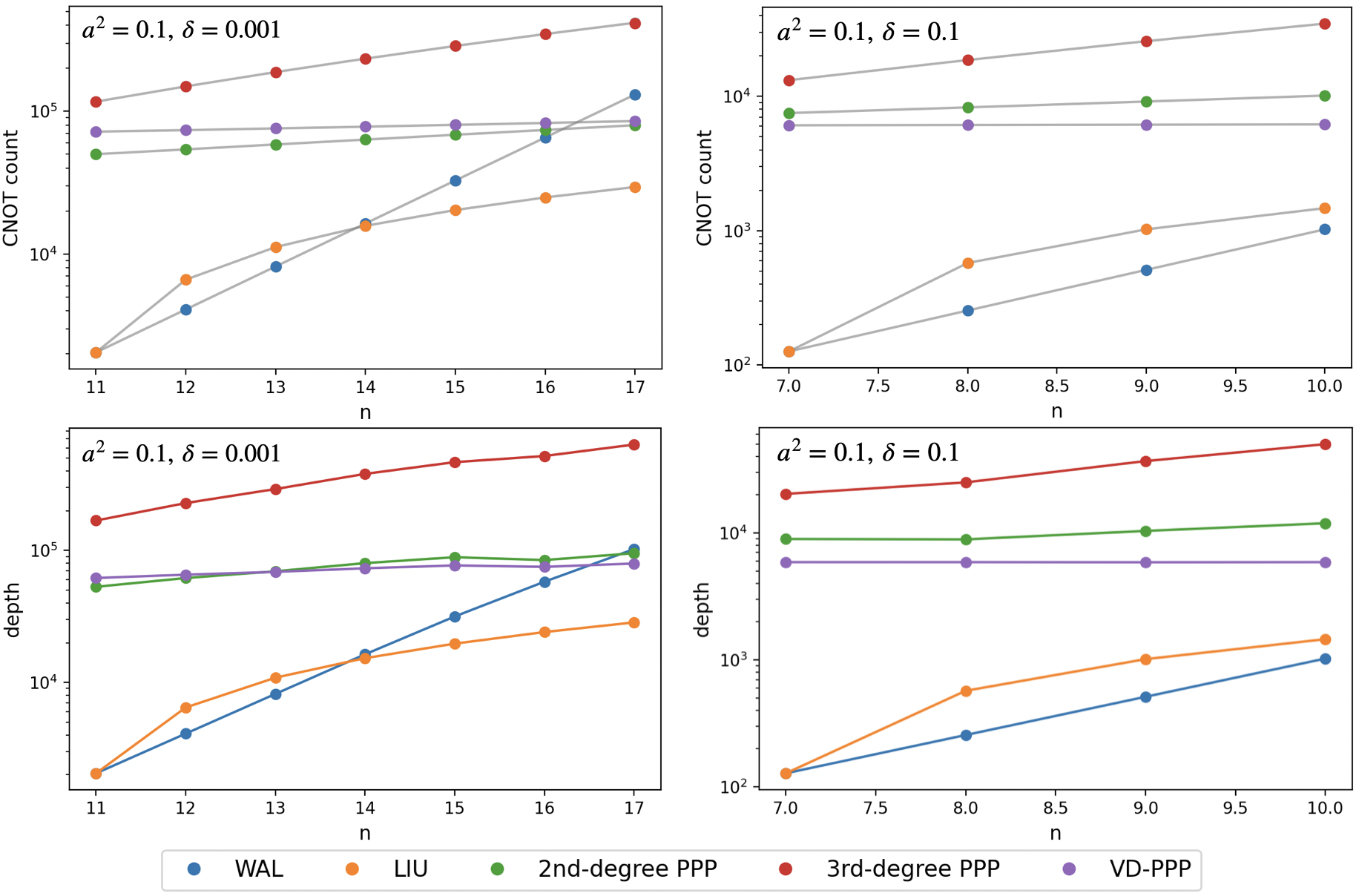}
}
\caption{CNOT count/depth by Qiskit with respect to $n$ for the modified Coulomb potential. The parameters $A=1$, $a^2=0.1$, and $L=20$ are constants. The first column corresponds to a relatively high precision $\delta=0.001$, while the second corresponds to a relatively low precision $\delta=0.1$. The gray lines in the subplots of the CNOT count are the guidelines of the theoretical CNOT count for each method discussed in Sect.~\ref{sec:comp}. 2nd-degree PPP and 3rd-degree PPP denote PPP by Algorithm 1 for $p=2$ and $p=3$, respectively, while VD-PPP denotes PPP by Algorithm 2 for $\mathcal{P}=\{1,2,3\}$ and $\mathcal{M}=\{2,\ldots, n\}$ with the objective function taking the theoretical CNOT count. } 
\label{appB:Fig3}
\end{figure}
Except for PPP with $p=3$ (i.e., cubic Hermite spline), the theoretical CNOT count coincides with the CNOT count by Qiskit after decomposition into the native gates. Although it is possible to cancel a small part of CNOT gates when a high-degree spline is employed, in principle, it is sufficient to discuss the theoretical CNOT count. 
On the other hand, Fig.~\ref{appB:Fig3} also shows the circuit depth after decomposition into the native gates using Qiskit \cite{Qiskit2023}. For the lower precision $\delta=0.1$, WAL is the best one among the methods. In contrast, for the higher precision $\delta=0.001$, WAL is the best method for $n\le 13$, and LIU is the best for $n\ge 14$, which is the same observation as for CNOT count. 

We also investigate a damped oscillating function as Example \ref{exa:2} with parameters $A=1$, $a=0.1$, $\omega=2$, and $L=20$ under a relatively high precision $\delta=0.001$ and a lower one $\delta=0.1$. 
We find a similar observation as Example \ref{exa:1} in Fig.~\ref{appB:Fig4}, but small gaps between the theoretical CNOT counts and the CNOT counts by Qiskit for several methods, including 3rd-degree PPP. We believe these gaps come from the possible vanishing of some phase gates regarding a specific function, and it does not influence the comparison result whether we use the theoretical gate counts or the ones by Qiskit emulation. 
\begin{figure}
\centering
\resizebox{13cm}{!}{
\includegraphics[keepaspectratio]{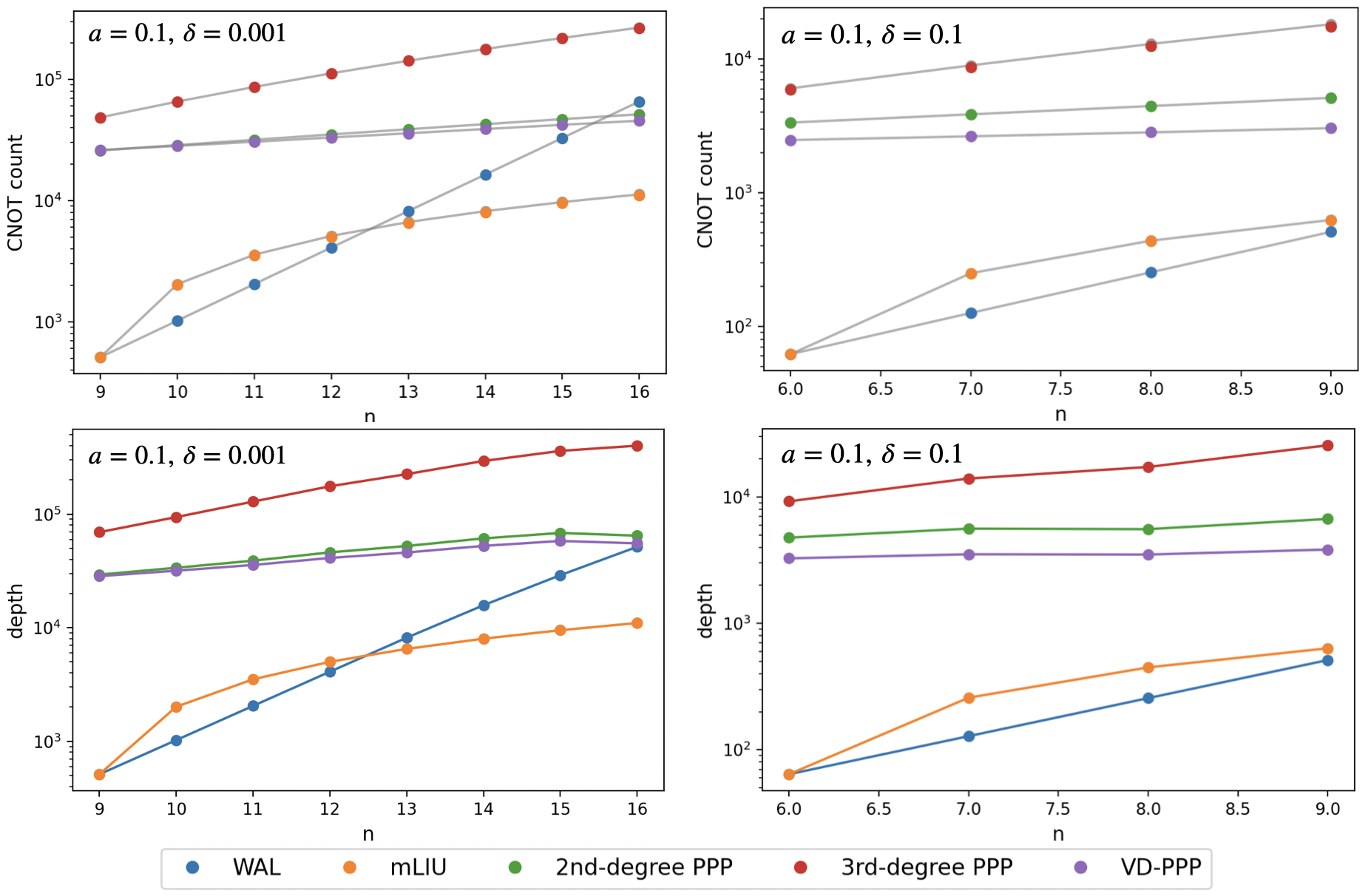}
}
\caption{CNOT count/depth by Qiskit with respect to $n$ for a damped oscillating function. The parameters $A=1$, $a=0.1$, $\omega=2$, and $L=20$ are constants. The first column corresponds to a relatively high precision $\delta=0.001$, while the second corresponds to a relatively low precision $\delta=0.1$. The gray lines in the subplots of the CNOT count are the guidelines of the theoretical CNOT count for each method discussed in Sect.~\ref{sec:comp}. 2nd-degree PPP and 3rd-degree PPP denote PPP by Algorithm 1 for $p=2$ and $p=3$, respectively, while VD-PPP denotes PPP by Algorithm 2 for $\mathcal{P}=\{1,2,3\}$ and $\mathcal{M}=\{2,\ldots, n\}$ with the objective function taking the theoretical CNOT count. } 
\label{appB:Fig4}
\end{figure}

If we focus on the circuit depth in Figs.~\ref{appB:Fig3}--\ref{appB:Fig4}, then we find an interesting observation that the depth for 2nd-degree PPP has a relatively small value when $n=2^k$ for some $k\in \mathbb{N}$. Since this is not observed for CNOT count, it indicates an efficient parallelization of 2nd-degree PPP as $n=2^k$. 

\section{Application to first-quantized Hamiltonian simulation}
\label{sec:appC}

\subsection{Linear depth for the part of two-particle interaction potential}
\label{subsec:appC0}

We prove that in the first-quantized Hamiltonian simulation for $N_e$ particles, the circuit depth for the real-time evolution of the two-particle interaction potential is $O(N_e)$, though the number of interaction potentials is $C^{N_e}_2 = \frac{1}{2}N_e(N_e-1)$. 
In fact, we can also include the one-particle potential and keep the linear depth with respect to the number of particles. 

We characterize the one-particle and two-particle potentials by unordered pairs $\{(i,j)\}_{i,j=1,\ldots,N_e}$. Here, $(i,j)$ describes the interaction potential between the $i$-th particle and the $j$-th particle, and indicates the one-particle potential for the $i$-th particle if $j=i$. 

In this formulation, our target is to seek $N_e$ sets, the pairs are implemented simultaneously in each set, such that they includes all the one-particle and two-particle potentials once and only once and any particle is referred only once in each set. 

Without geometric constraints, the construction of the sets are not unique in general. Here, we provide only one candidate, which can be readily verified. We introduce the sets for an odd number, and then we use the introduced sets to define the ones for an even number. 
Henceforth, we denote the number of particles by $N$ for simplicity. 
\begin{itemize}
\item For $N = 2n-1$, $n\in \mathbb{N}$, define the $k$-th set by 
$$
[k]_{N}^{\text{odd}} := \left\{([k-j]_N,[k+j]_N): j=0,\ldots,n-1\right\}, \quad k=1,\ldots,N.
$$

\item For $N = 2n$, $n\in \mathbb{N}$, define the $k$-th set by 
\begin{align*}
[k]_{N}^{\text{even}} &:= [k]_{N-1}^{\text{odd}} \setminus \{(k,k)\} \cup \{(k,N)\}, \quad k=1,\ldots,N-1, \\
[N]_{N}^{\text{even}} &:= \{(j,j): j=1,\ldots,N\}. 
\end{align*}
\end{itemize}
Here, $[p]_N\in \{1,\ldots,N\}$ denotes $p \mod N$. For the convenience, we define $[N]_N = N$ instead of $[N]_N = 0$. By the definition of the sets, it is clear that for $N=2n-1$, $n\in \mathbb{N}$, each set includes $n$ pairs, and hence, $2n$ numbers are referred in total. 
In order to achieve our target, we need to verify the properties in the previous paragraph. In other words, we prove the following theorem. 
\begin{theorem}
Let $N=2n-1$, $n\in \mathbb{N}$. 

\noindent $\mathrm{(i)}$ In each $[k]_{N}^{\text{odd}}$, $k=1,\ldots,N$, except for the pair $(k,k)$, any number in $\{1,\ldots,N\}\setminus\{k\}$ appears once and only once. This implies that different elements in $[k]_{N}^{\text{odd}}$ do not share a common number. 

\noindent $\mathrm{(ii)}$ For $k, k^\prime=1,\ldots,N$ and $k\not=k^\prime$, we have $[k]_N^{\text{odd}} \cap [k^\prime]_N^{\text{odd}} = \emptyset$. 
\end{theorem}
\begin{proof}
\noindent $\mathrm{(i)}$ Note that for each $k=1,\ldots,N$, the set $[k]_N^{\text{odd}}$ contains exact $n$ pairs and $2n$ numbers including $k$ twice. Thus, ``any number in $\{1,\ldots,N\}$ appears" is equivalent to ``any number in $\{1,\ldots,N\}\setminus\{k\}$ appears once and only once". 
Then, it is sufficient to prove that for each $\ell=1,\ldots,N$, and $k=1,\ldots,N$, there exists $j_{k\ell}\in \{0,\ldots,n-1\}$ such that either $[k-j_{k\ell}]_N = \ell$ or $[k+j_{k\ell}]_N = \ell$. 
This is verified by the following explicit choice of $j_{k\ell}$: 
$$
j_{k\ell} = \left\{
\begin{aligned}
& k-\ell, && 0\le k-\ell \le n-1, \\
& N-(k-\ell), && n\le k-\ell \le 2n-2, \\
& \ell-k, && -n+1\le k-\ell \le -1, \\
& N-(\ell-k), && -2n+1 \le k-\ell \le -n.
\end{aligned}
\right.
$$

\noindent $\mathrm{(ii)}$ We prove by contradiction. Without loss of generality, we assume that there exists two integers $k>k^\prime$ such that $[k]_N^{\text{odd}} \cap [k^\prime]_N^{\text{odd}} \not= \emptyset$. According to the definition of the sets, we have either
$$
\left\{
\begin{aligned}
& [k-j]_N = [k^\prime-j^\prime]_N, \\
& [k+j]_N = [k^\prime+j^\prime]_N,
\end{aligned}
\right.
\quad\text{or }\quad
\left\{
\begin{aligned}
& [k-j]_N = [k^\prime+j^\prime]_N, \\
& [k+j]_N = [k^\prime-j^\prime]_N,
\end{aligned}
\right.
$$
for some integers $j,j^\prime=0,\ldots,n-1$. In both cases, we obtain 
$$
[k-j]_N + [k+j]_N = [k^\prime+j^\prime]_N + [k^\prime-j^\prime]_N,
$$
that is, 
$$
k-j+k+j + mN = k^\prime-j^\prime+k^\prime+j^\prime + m^\prime N,  
$$
equivalently, 
$$
2(k-k^\prime) = (m^\prime-m) N,
$$
for some $m,m^\prime\in \mathbb{Z}$. 
Since $N=2n-1$ is an odd number, we conclude that $m^\prime-m$ is even, and hence there exists $\tilde m\in \mathbb{Z}$ such that $k-k^\prime=\tilde m N$. This implies $k=k^\prime$, which is a contradiction. 
\end{proof}

On the other hand, for $N=2n$, $n\in\mathbb{N}$, by the definition of the sets, it is readily to see that each set (except the last one) includes $n$ pairs, and hence, $2n$ numbers are referred in total, and the last one includes $N$ pairs: $\{(j,j)\}_{j=1,\ldots,N}$. We employ the above theorem to prove the following theorem for an even number. 
\begin{theorem}
Let $N=2n$, $n\in \mathbb{N}$. 

\noindent $\mathrm{(i)}$ In each $[k]_{N}^{\text{even}}$, $k=1,\ldots,N-1$, any number in $\{1,\ldots,N\}$ appears once and only once. This implies that different elements in $[k]_{N}^{\text{even}}$ do not share a common number. 

\noindent $\mathrm{(ii)}$ For $k, k^\prime=1,\ldots,N$ and $k\not=k^\prime$, we have $[k]_N^{\text{even}} \cap [k^\prime]_N^{\text{even}} = \emptyset$. 
\end{theorem}
\begin{proof}
\noindent $\mathrm{(i)}$ In $[N]_{N}^{\text{even}}$, different elements do not share a common number by the definition. For $k=1,\ldots,N-1$, by Theorem 1(i), any number in $\{1,\ldots,N-1\}\setminus \{k\}$ appears once and only once in $[k]_{N-1}^{\text{odd}}\setminus \{(k,k)\}$. Therefore, we find that any number in $\{1,\ldots,N\}$ appears once and only once in $[k]_{N}^{\text{even}}$ by the definition. 

\noindent $\mathrm{(ii)}$ If $k=N$ (or $k^\prime=N$), then it is trivial because $[k^\prime]_N^{\text{even}}$ (or $[k]_N^{\text{even}}$) has no pairs with the same number for any $k^\prime=1,\ldots,N-1$. For $1\le k^\prime\not=k\le N-1$, by Theorem 1(ii), we have $[k]_{N-1}^{\text{odd}} \cap [k^\prime]_{N-1}^{\text{odd}} = \emptyset$. By the definition of $[k]_N^{\text{even}}$, it is sufficient to check $(k,N)\notin [k^\prime]_{N-1}^{\text{odd}}$ for any $1\le k^\prime \not= k\le N-1$, and $(k,N)\not=(k^\prime,N)$. These are trivial because there are no pairs with number $N$ in $[k^\prime]_{N-1}^{\text{odd}}$, and $k\not=k^\prime$. 
\end{proof}

\subsection{Implementation of distance register}
\label{subsec:appC1}

Recall that $\Delta x_i = L_i/2^n$, $i=1,\ldots,d$. For simplicity, we choose the same values for each dimension: $L_i=L$ and $\Delta x = L/2^n$. 
The distance in one dimension is the difference between two integers, which can be efficiently realized by a quantum comparator, a Fredkin gate, and an in-place modular subtractor. Thus, we only consider the case of $d\ge 2$ in the following contexts. 

\noindent \underline{Electron-nucleus distance}
We assume that the known position vector of the $v$-th nucleus is expressed by
$$
\bm{R}_v = \tilde{\bm{k}}_v \Delta x = \sum_{i=1}^d \Delta x\tilde k_{v,i} \bm{e}_i, 
$$
where $\tilde{\bm{k}}_v = (\tilde k_{v,1}, \ldots, \tilde k_{v,d})^T$, $v=1,\ldots, N_{\text{nuc}}$ is an integer-valued vector. Then, the operation $U_{\text{dis},v}(\bm{R}_v)$ aims at
\begin{align}
\label{app:eq-C1}
U_{\text{dis},v}(\bm{R}_v) \left(\ket{\bm{k}_\ell}_{dn} \otimes \ket{0}_{n_{\text{anc}}}\right) = \ket{r_{\ell}(\bm{R}_v)}_{n_{\text{dis}}} \otimes \ket{\bm{q}_{\ell,v}}_{dn+n_{\text{anc}}-n_{\text{dis}}}, 
\end{align}
for $\ell=1,\ldots, N_{\text{e}}$. Here, $\ket{r_{\ell}(\bm{R}_v)}$ is the electron-nucleus distance register, and $r_{\ell}(\bm{R}_v) := \sum_{i=1}^d (k_{\ell,i}-\tilde k_{v,i})^2$. 
To be precise, this should be called a squared distance register because we have
$$
r_{\ell}(\bm{R}_v) = |\bm{r}_\ell - \bm{R}_v|^2/(\Delta x)^2. 
$$
In some papers, the distance register is defined as $\ket{\sqrt{|\bm{r}_\ell - \bm{R}_v|^2}}$ instead, but such a non-integer register needs extra discussion on the error, which we try to avoid in this paper. A non-integer distance register is compatible with the method APK, and we should discuss it in detail in future work. 
In Eq.~\eqref{app:eq-C1}, $\ket{\bm{q}_{\ell,v}}$ is a quantum state depending on $\bm{k}_\ell$ and $\tilde{\bm{k}}_v$, and $n_{\text{anc}}$ is the number of ancillary qubits to achieve such an operation. Moreover, it is enough to take $n_{\text{dis}} = 2n +\lceil\log_2 d\rceil$ according to the definition of $r_{\ell}(\bm{R}_v)$. 

\noindent \underline{Electron-electron distance}
An electron-electron distance register $\ket{r_{\ell,\ell^\prime}}$ is similarly defined using $r_{\ell,\ell^\prime} := \sum_{i=1}^d (k_{\ell,i}-k_{\ell^\prime,i})^2$, and $U_{\text{dis}}$ is an operation aiming at 
\begin{align*}
U_{\text{dis}} \left(\ket{\bm{k}_{\ell^\prime}}_{dn} \otimes \ket{\bm{k}_{\ell}}_{dn} \otimes \ket{0}_{n_{\text{anc}}}\right) = \ket{r_{\ell,\ell^\prime}}_{n_{\text{dis}}} \otimes \ket{\bm{q}_{\ell,\ell^\prime}}_{2dn+n_{\text{anc}}-n_{\text{dis}}}.
\end{align*}

The operations $U_{\text{dis},v}(\bm{R}_v)$ and $U_{\text{dis}}$ can be realized by arithmetic operations, including absolute subtraction, square, and addition, as shown in Fig.~\ref{appC:Fig1}. 
For each $i=1,\ldots,d$, the absolute subtractions $U_{\text{sub},v}^{(i)}$ and $U_{\text{sub}}$ can be constructed by the circuits in Fig.~\ref{appC:Fig2}, and the square operation $U_{\text{sqr}}$ is given in Fig.~\ref{appC:Fig3}. 
\begin{figure}
\centering
\resizebox{15cm}{!}{
\includegraphics[keepaspectratio]{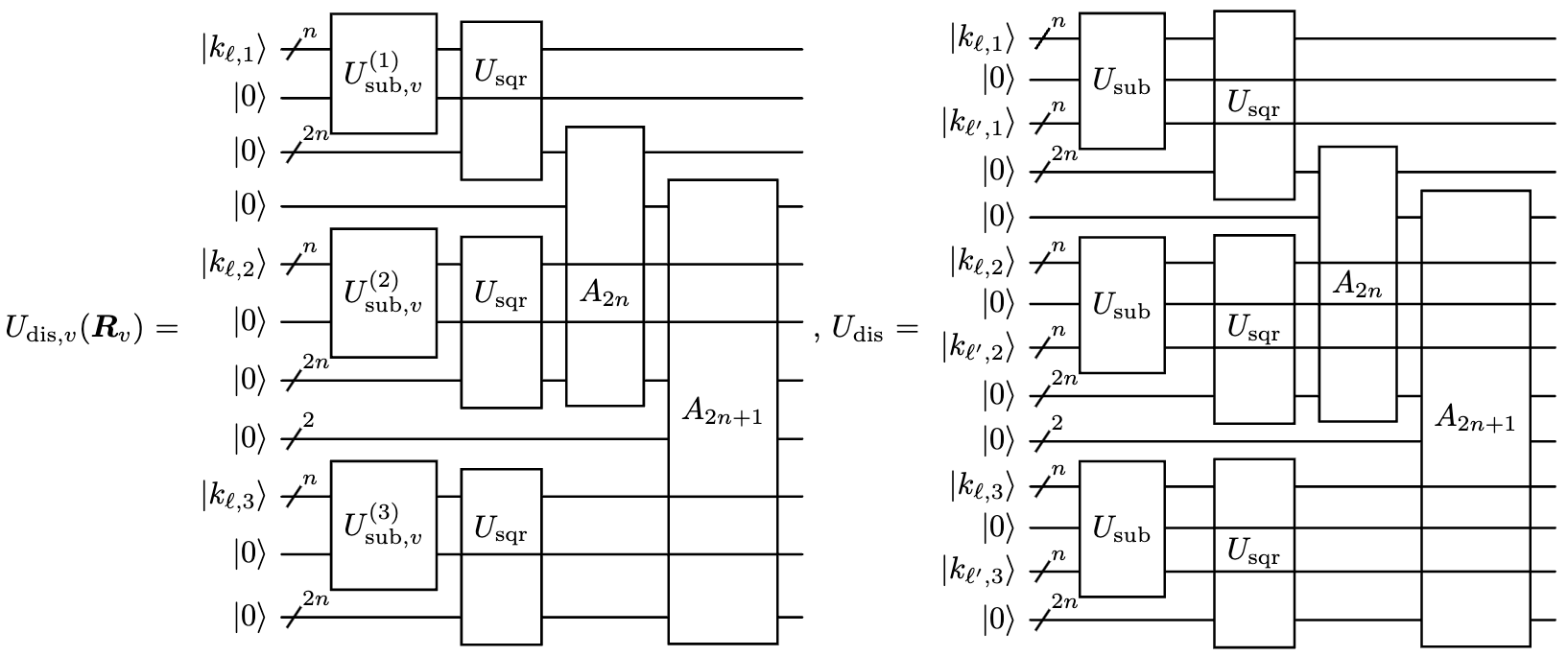}
}
\caption{Implementations of $U_{\text{dis},v}(\bm{R}_v)$ and $U_{\text{dis}}$ by using arithmetic operations for $d=3$. Here, $U_{\text{sub},v}^{(i)}$ and $U_{\text{sub}}$ are quantum absolute value subtractors, $U_{\text{sqr}}$ is a quantum circuit of squaring, and $A_k$ denotes a quantum adder of two $k$-qubit registers with a carry qubit in the center.}  
\label{appC:Fig1}
\end{figure}
\begin{figure}
\centering
\resizebox{15cm}{!}{
\includegraphics[keepaspectratio]{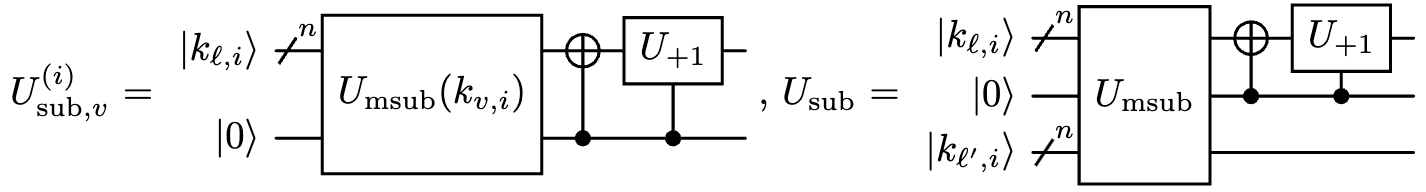}
}
\caption{Implementations of $U_{\text{sub},v}^{(i)}$ and $U_{\text{sub}}$ by modular subtractors $U_{\text{msub}}(\tilde k_{v,i})$, $U_{\text{msub}}$, CNOT gates, and a controlled increment gate. Here, the CNOT gate denotes $n$ CNOT gates controlled by the ancillary qubit and targeted at the $n$ qubits in the first register. } 
\label{appC:Fig2}
\end{figure}
\begin{figure}
\centering
\resizebox{10cm}{!}{
\includegraphics[keepaspectratio]{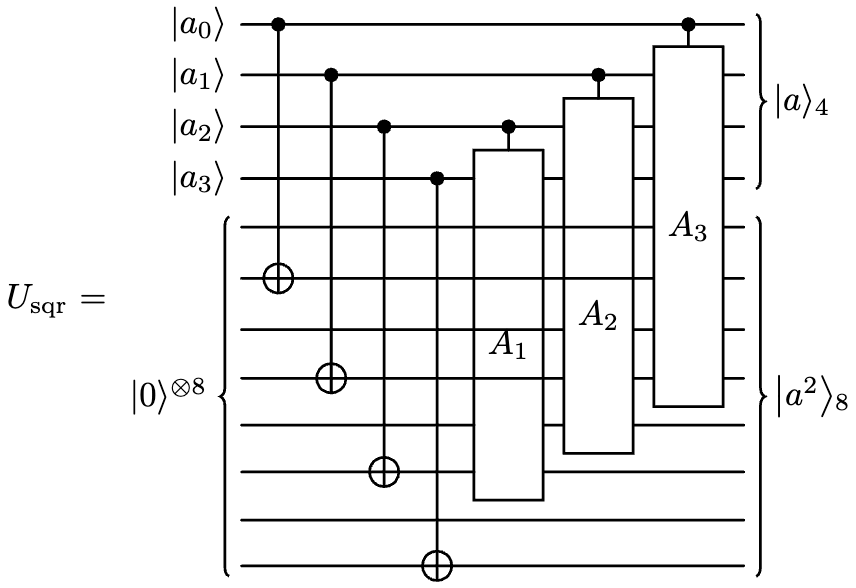}
}
\caption{Implementation of $U_{\text{sqr}}$ based on schoolbook method in an example of $n=4$. Here, $A_k$ denotes the same quantum adder as used in Fig.~\ref{appC:Fig1}.} 
\label{appC:Fig3}
\end{figure}
In terms of the implementation in Fig.~\ref{appC:Fig1}, we have $n_{\text{anc}} = d+2nd+d(d-1)/2 = 2dn + d(d+1)/2$. 
Moreover, if we employ the (controlled) modular adder/subtractor in \cite{Li.2020}, the modular subtractor with a given integer in \cite{Yuan.2023}, and the controlled increment gate in \cite{Yuan.2023}, then the gate count is $\mathcal{O}(dn^2)$ with depth $\mathcal{O}(n^2+dn)$. 
In the case of $d=2$, instead of Fig.~\ref{appC:Fig1}, one can alternatively employ a subtractor with a reversible squaring and sum-of-squares unit in \cite{Nagamani.2018} using more ancillary qubits of order $\mathcal{O}(n^2)$. 
 
On the other hand, as we try to minimize the number of ancillary qubits, we provide alternative implementations of operations $U_{\text{dis},v}(\bm{R}_v)$ and $U_{\text{dis}}$ using QFTs and controlled polynomial phase gates. The previous papers \cite{Yuan.2023, Draper2000, Ruiz-Perez.2017, Sahin2020} applied QFTs and phase gates to construct the quantum adder, multiplier, etc. 
Here, we use a similar idea and propose circuits deriving a second-degree polynomial in the register, as shown in Fig.~\ref{appC:Fig4}. 
\begin{figure}
\centering
\resizebox{15cm}{!}{
\includegraphics[keepaspectratio]{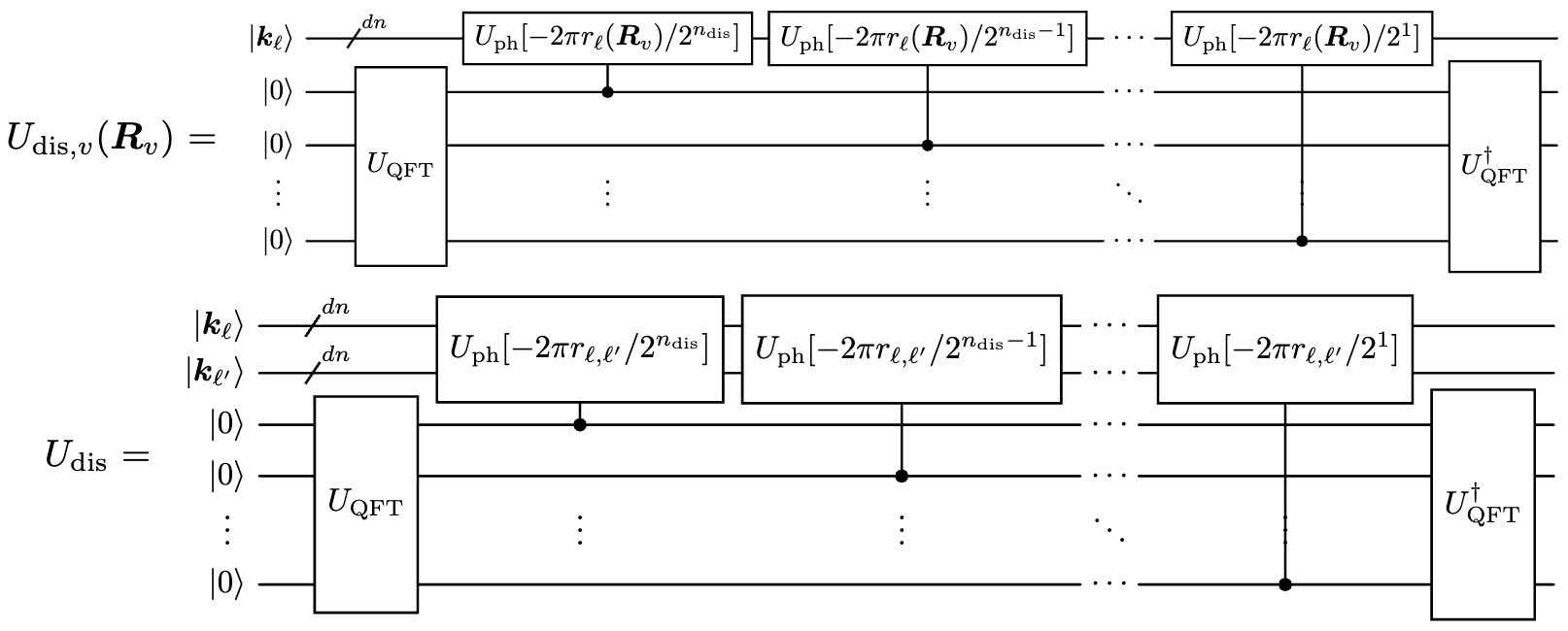}
}
\caption{Alternative implementations of $U_{\text{dis},v}(\bm{R}_v)$ and $U_{\text{dis}}$ using QFTs and controlled polynomial phase gates. Here, $U_{\text{ph}}$ denotes the polynomial phase gate defined in Appendix \ref{subsec:appA3}. The number of zero-initialized ancillary qubits is $2n+\lceil\log_2 d\rceil$. } 
\label{appC:Fig4}
\end{figure}
These circuits use only $n_{\text{anc}}=n_{\text{dis}}=2n+\lceil\log_2 d\rceil$ ancillary qubits, which is smaller than those in Fig.~\ref{appC:Fig1}.  
We verify the above circuit for $U_{\text{dis},v}(\bm{R}_v)$ by the definition of polynomial phase gate. For simplicity, we denote $n_{\text{dis}}$ by $m$ and have 
\begin{align*}
\ket{0}^{\otimes m} \otimes \ket{\bm{k}_\ell}_{dn} &\xrightarrow{U_{\text{QFT}}} \frac{1}{\sqrt{2^{m}}} \sum_{q_0,\ldots,q_{m-1}=0}^1 \ket{q_{m-1}} \otimes \cdots \otimes \ket{q_0} \otimes \ket{\bm{k}_\ell}_{dn} \\
&\xrightarrow{\text{C}U_{\text{ph}}} \frac{1}{\sqrt{2^{m}}} \sum_{q_0,\ldots,q_{m-1}=0}^1 \mathrm{exp}\left(\mathrm{i}2\pi q_0 \sum_{i=1}^d (k_{\ell,i}-\tilde k_{v,i})^2/2^m\right) \ket{q_{m-1}} \otimes \cdots \otimes \ket{q_0} \otimes \ket{\bm{k}_\ell}_{dn} \\
&\xrightarrow{\text{C}U_{\text{ph}}} \frac{1}{\sqrt{2^{m}}} \sum_{q_0,\ldots,q_{m-1}=0}^1 \mathrm{exp}\left(\mathrm{i}2\pi (q_0+2q_1) \sum_{i=1}^d (k_{\ell,i}-\tilde k_{v,i})^2/2^m\right) \ket{q_{m-1}} \otimes \cdots \otimes \ket{q_0} \otimes \ket{\bm{k}_\ell}_{dn} \\
&\xrightarrow{\text{C}U_{\text{ph}}} \cdots \xrightarrow{\text{C}U_{\text{ph}}} \frac{1}{\sqrt{2^{m}}} \sum_{q=0}^{2^m-1} \mathrm{exp}\left(\mathrm{i}2\pi q \sum_{i=1}^d (k_{\ell,i}-\tilde k_{v,i})^2/2^m\right) \ket{q}_m \otimes \ket{\bm{k}_\ell}_{dn} \\
&\xrightarrow{U_{\text{QFT}}^\dag} \frac{1}{2^{m}} \sum_{q,q^\prime=0}^{2^m-1} \mathrm{exp}\left(\mathrm{i}2\pi q \left(\sum_{i=1}^d (k_{\ell,i}-\tilde k_{v,i})^2-q^\prime\right)/2^m\right) \ket{q^\prime}_m \otimes \ket{\bm{k}_\ell}_{dn} \\
&= \ket{\sum_{i=1}^d (k_{\ell,i}-\tilde k_{v,i})^2}_m \otimes \ket{\bm{k}_\ell}_{dn}.
\end{align*}
The last equation follows from the identity $\sum_{q=0}^{2^m-1} \mathrm{exp}\left(\mathrm{i}2\pi q \tilde q/2^m\right)= 2^m \delta_{\tilde q,0}$ provided that $\tilde q$ is an integer in $[-2^m+1,2^m-1]$. The verification for $U_{\text{dis}}$ is similar. 

Now, we discuss the detailed gate count and the circuit depth of Fig.~\ref{appC:Fig4}. We note that the polynomial phase gate $U_{\text{ph}}[-2\pi r_{\ell}(\bm{R}_v)/2^m] = U_{\text{ph}}\left[-2\pi \sum_{i=1}^d (k_{\ell,i}-\tilde k_{v,i})^2/2^m\right]$ can be implemented by $d$ times phase gates for quadratic functions $-2\pi (x_i-\tilde k_{v,i})^2/2^m$ on $n$ qubits, and $U_{\text{QFT}}$ on a zero register is equivalent to the Hadamard gates. Thus, $U_{\text{dis},v}(\bm{R}_v)$ uses a global phase gate, $4n+2\lceil\log_2 d\rceil$ Hadamard gates, $dn(2n+\lceil\log_2 d\rceil)+(2n+\lceil\log_2 d\rceil)(2n+\lceil\log_2 d\rceil-1)/2$ controlled phase gates, and $dn(n-1)(2n+\lceil\log_2 d\rceil)/2$ $2$-controlled phase gates with depth $1+8(2n+\lceil\log_2 d\rceil)-10+n(6n+2)\max\{d,2n+\lceil\log_2 d\rceil\}$. 
Here, the gate count and the depth of Fig.~\ref{appC:Fig4} are both $\mathcal{O}(n^3)$, one order higher than the circuit in Fig.~\ref{appC:Fig1} due to the limited number of ancillary qubits. As for the operation $U_{\text{dis}}$, if we regard $\ket{k_{\ell^\prime,i}}_n \otimes \ket{k_{\ell,i}}_n$ as a new register $\ket{\hat k_{\ell,\ell^\prime,i}}_{2n}$, then the polynomial phase gate $U_{\text{ph}}[-2\pi r_{\ell,\ell^\prime}/2^m]$ can also be implemented by $d$ times phase gates for quadratic functions on $2n$ qubits. Thus, $U_{\text{dis}}$ uses a global phase gate, $4n+2\lceil\log_2 d\rceil$ Hadamard gates, $2dn(2n+\lceil\log_2 d\rceil)+(2n+\lceil\log_2 d\rceil)(2n+\lceil\log_2 d\rceil-1)/2$ controlled phase gates, and $dn(2n-1)(2n+\lceil\log_2 d\rceil)$ $2$-controlled phase gates with depth $1+8(2n+\lceil\log_2 d\rceil)-10+2n(12n+2)\max\{d,2n+\lceil\log_2 d\rceil\}$.

\end{document}